\documentclass[a4paper,10pt]{article}

\usepackage[margin=1.25in]{geometry}

\usepackage{cite}    

\usepackage{tikz}
\usepackage{pgffor}
\usepackage[utf8]{inputenc}
\usepackage{framed}
\usepackage{boxedminipage}
\usepackage{graphicx}
\usepackage{amsmath,amssymb,amsthm}
\usepackage[textsize=tiny,textwidth=2.2cm,shadow]{todonotes}
\usepackage{caption}
\usepackage{subcaption}

\usetikzlibrary{patterns}

\newtheorem{theorem}{Theorem}
\newtheorem{definition}[theorem]{Definition}
\newtheorem{observation}[theorem]{Observation}
\newtheorem{corollary}[theorem]{Corollary}
\newtheorem{lemma}[theorem]{Lemma}

\newcommand{\eps}{\varepsilon}
\newcommand{\schedule}{\ensuremath{\mathcal{S}}}
\newcommand{\abc}[3]{\ensuremath{\textup{#1}|\,#2\,|#3}}
\newcommand{\OO}{\mathcal{O}}
\newcommand{\UB}{\text{UB}}
\newcommand{\opt}{\textsc{Opt}}
\newcommand{\alg}{\textsc{Alg}}

\newcommand{\x}{\ensuremath{\tau}} 
\newcommand{\ka}{\ensuremath{\kappa}} 
\newcommand{\w}{\ensuremath{a}} 

\def\e{{\varepsilon}}

\def\OPT{\text{OPT}}
\def\APPROX{\text{APX}}

 \newenvironment{algorithmNB}[1]%
{
\begin{center}
\begin{minipage}{\textwidth}
\rule{\textwidth}{0.5pt}\\
\noindent{\textbf{Algorithm \textsc{#1}}}

\noindent\ignorespaces
}
{
\smallskip
\vspace*{-2ex}
\end{minipage}

\rule[1ex]{\textwidth}{0.5pt}
\end{center}
}

\newcommand{\jcom}[1]{} 
\newcommand{\ncom}[1]{} 

\usepackage{color}

\title{Dual techniques for scheduling\\ on a machine with varying speed\footnote{Parts of the results appeared in a preliminary version of this paper in the proceedings of \textit{ICALP '13}~\cite{MegowV13}.}}

\author{Nicole Megow\thanks{Department of Mathematics, Technische Universit{\"a}t
       Berlin, Germany.  Email: \texttt{nmegow@math.tu-berlin.de}. Supported by the German Science Foundation (DFG) under contract  ME 3825/1.}%
     \and Jos\'e Verschae\thanks{Departamento de Ingenier\'ia Industrial and Centro de Modelamiento Matem\'atico, Universidad de Chile, Santiago, Chile. Email: {\tt jverscha@ing.uchile.cl}. Supported by the Nucleo Milenio Informaci\'on y Coordinaci\'on en Redes ICM/FIC P10-024F and by FONDECYT project 3130407.}
}
\date{\today}

\begin{document}

  \maketitle

\begin{abstract}
We study scheduling problems on a machine with varying speed. Assuming a known speed function we ask for a cost-efficient scheduling solution. Our main result is a PTAS for minimizing the total weighted completion time in this setting. This also implies a PTAS for the closely related problem of scheduling to minimize generalized global cost functions. The key to our results is a re-interpretation of the problem within the well-known {\em two-dimensional Gantt chart}: instead of the standard approach of scheduling in the {\em time-dimension}, we construct scheduling solutions in the {\em weight-dimension}. 

We also consider a dynamic problem variant in which deciding upon the speed is part of the scheduling problem and we are interested in the tradeoff between scheduling cost and speed-scaling cost, which is typically the energy consumption. We observe that the optimal order is independent of the energy consumption and that the problem can be reduced to the setting where the speed of the machine is fixed, and thus admits a PTAS. Furthermore, we provide an FPTAS for the NP-hard problem variant in which the machine can run only on a fixed number of discrete speeds. Finally, we show how our results can be used to obtain a~$(2+\eps)$-approximation for scheduling preemptive jobs with release dates on multiple identical parallel machines.
\end{abstract}

\paragraph{Key words:} scheduling, speed-scaling, power-management, generalized cost functions, non-availability periods, energy-aware

\section{Introduction}

In several computation and production environments we face scheduling problems in which the speed of resources may vary. We distinguish mainly two types of varying-speed scenarios: one in which the speed is a \emph{given} function of time and another \emph{dynamic} setting in which deciding upon the processor speed is part of the scheduling problem. The first setting occurs, e.g., in production environments where the speed of a resource may change due to overloading, aging, or in an extreme case it may be completely unavailable due to maintenance or failure. 
The dynamic setting finds application particularly in modern computer architectures, where speed-scaling is an important tool for power-management. Here we are interested in the tradeoff between the power consumption and the quality-of-service. Both research directions---scheduling on a machine with given speed fluctuation as well as scheduling including speed-scaling---have been pursued quite extensively,  but seemingly separately from each other. 

The main focus of our work and the main technical contribution lie in the setting with a given speed function. We consider the problem of scheduling to minimize the sum of weighted completion times~$\sum_j w_jC_j$, a standard measure for quality-of-service. We present a PTAS for this problem which is best possible unless P$=$NP. 
In addition, we draw an interesting connection to the dynamic model which allows us to transfer some of our techniques to this setting.

Very useful in our arguments is the well-known geometric view of the min-sum scheduling
problem in a {\em two-dimensional Gantt chart}, an interpretation originally
introduced by Eastman, Even, and Isaacs~\cite{EEI64}. Crucial to our results is the deviation from the standard view of scheduling in the {\em time dimension} and switching to scheduling in the {\em weight dimension}. This dual view allows us to cope with the highly sensitive speed changes in the time dimension which prohibit standard rounding, guessing, and approximation techniques.

\subsection*{Previous work}

Research on scheduling on a machine of given varying speed has mainly focused on the special case of scheduling with non-availability periods, see e.g.~\cite{schmidt00,lee04,diedrichJST09,maCZ10}. Despite a history of more than 30 years, only recently the first constant approximation for $\min \sum w_jC_j$ was derived by Epstein et al.~\cite{epsteinLMMMSS12}. In fact, their $(4+\eps)$-approximation computes a universal sequence which has the same guarantee for any (unknown) speed function. For the setting with release dates, they give an approximation algorithm with the same guarantee for any given speed function. If the speed is only non-decreasing (and the release dates are trivial), there is an efficient PTAS~\cite{stillerW10}. In this case the complexity status remains open, whereas for general speed functions the problem is strongly NP-hard, even when for each job the weight and processing time are~equal~\cite{wangSC05}.

The problem of scheduling on a machine of varying speed is equivalent to scheduling on an ideal machine (of constant speed) but minimizing a more general global cost function~$\sum w_jf(C_j)$, where~$f$ is a nondecreasing function. In this identification, $f(C)$ denotes the time that the varying-speed machine needs to process a work volume of $C$~\cite{hoehnJ12}. The special case of only nondecreasing (nonincreasing) speed functions corresponds to concave (convex) global cost functions. In a recent work, H\"ohn and Jacobs~\cite{hoehnJ12} give a formula for computing tight guarantees for Smith's rule for any convex or concave function~$f$. They also show that the problem for increasing piecewise linear cost function is strongly NP-hard even with only two different slopes, and so is our problem when the speed function takes only two~distinct~values.   

Even more general min-sum cost functions have been studied, where each job may have its individual nondecreasing
cost function. 
For this setting, the currently best known approximation factor is~$4+\eps$~\cite{cheungS11,MestreV14}. For the more complex setting with release dates, Bansal and Pruhs~\cite{bansalP10} give a randomized $\OO(\log \log(n \max_j p_j))$-approximation algorithm. Clearly, these results translate also to the setting with varying machine speed.

Scheduling with dynamic speed-scaling was initiated by Yao, Demers, and Shenker~\cite{yaoDS95} and became a very active research field in the past fifteen years. Most work focuses on scheduling problems where jobs have deadlines by which they must finish. Thus, the speed scaling problem is a single-objective minimization problem. We refer to~\cite{iraniP05,albers10} for an overview. Closer to our setting 
is the work initiated by Pruhs, Uthaisombut, and Woeginger~\cite{pruhsUW08} where they obtain a polynomial algorithm for minimizing the total flow time given an energy budget if all jobs have the same work volume. This work is later continued by many others; see, e.\,g.,~\cite{albers_energy-efficient_2007,bansalPS09,chan_non-clairvoyant_2010} and the references therein. Most of this literature is concerned with online algorithms to minimize total (or weighted) flow time plus energy. 
The minimization of the weighted sum of completion times plus energy has been considered recently.
Angel, Bampis, and Kacem~\cite{angelBK12} derive constant approximations for non-preemptive models with unrelated machines and release dates. Carrasco, Iyengar, and Stein~\cite{carrascoIS11} obtain similar results even under precedence constraints. 

For the general objective of speed-scaling with an energy budget as considered in this paper, Angel, Bampis, and Kacem~\cite{angelBK12} also show a randomized~$(2+\eps)$-approximation slightly exceeding the energy budget for non-preemptive scheduling on unrelated machines with release dates. The bounds given for scheduling cost and the budget excess are satisfied only in expectation. 

\subsection*{Our results} 
We give several best possible algorithms for problem variants that involve scheduling to minimize the total weighted completion time on a single machine that may vary its speed.

Our main result is an efficient PTAS (Section~\ref{sec:PTAS-given-speeds}) for scheduling to minimize $\sum w_jC_j$ on a machine of varying speed (given by an oracle). This is best possible since the problem is strongly NP-hard, even when the machine speed takes only two distinct values~\cite{hoehnJ12}. 
Our results generalize recent previous results such as a PTAS on a machine with only {\em non-decreasing} speeds~\cite{stillerW10} and FPTASes for only {\em one} non-availability period~\cite{kellererS10,kacemM09}. 

Standard scheduling techniques rely on delaying jobs or rounding processing requirements. Such approaches typically fail on varying-speed machines. The reason is that the slightest error introduced by rounding might provoke an unbounded increase in the solution cost.
Similarly, adding any amount of idle time to the machine might be fatal. Our
techniques completely avoid this difficulty by a change of paradigm. To explain
our ideas it is helpful to use a 2D-Gantt chart interpretation~\cite{EEI64}; see
Section~\ref{sec:weight-space}. As observed before, e.g., in~\cite{goemans_2D_2000}, we obtain a \emph{dual} scheduling
problem by looking at the y-axis in a 2D-Gantt chart and switching the roles of
the processing times and weights. In other words, a dual solution describes a
schedule by specifying the remaining weight of the system at the moment a job
completes. %
This simple idea avoids the difficulties on the time-axis and allows to combine old with new techniques for scheduling on the weight-axis. We remark that this result translates directly to the equivalent problem \abc{1}{}{\sum_j w_j f(C_j)} (with $f$ non-decreasing).

In case that an algorithm can set the machine at arbitrary speeds, 
we show in Section~\ref{sec:cont-speeds} that the optimal scheduling sequence is independent of the available energy. This follows by analyzing a convex program that models the optimal energy assignment for a given job permutation. A similar observation
was made independently by V\'asquez~\cite{vas12} in a game-theoretic setting.  We show that computing this universal optimal sequence corresponds to the problem of scheduling with a particular concave global cost function, 
which can be solved with our PTAS mentioned above, or with a PTAS for non-decreasing speed~\cite{stillerW10}.
Interestingly, this reduction relies again on a problem transformation from time-space to weight-space in the 2D-Gantt chart. For a given scheduling sequence, we give an explicit formula for computing the optimal energy (speed) assignment. Thus, we have a PTAS for speed-scaling and scheduling for a given energy budget. We remark that the complexity of this problem is open. 

In many applications, including most modern computer architectures, machines are only capable of using a given number of discrete power (speed) states. We provide in Section~\ref{sec:discrete-speed} an efficient PTAS for this complex scenario. This algorithm is again based on our techniques relying on dual schedules. Furthermore, we obtain a $(1+\eps)$-approximation of the Pareto frontier for the energy-cost bicriteria problem. On the other hand, we show that this problem is NP-hard even when there are only two speed states. We complement this result by giving an FPTAS for a constant number of available speeds. 

In Section~\ref{sec:parallelMachines} we consider a more complex scheduling problem in the speed-scaling setting: jobs have individual release dates and must be scheduled preemptively on~$m$ identical parallel machines. We notice that our PTAS results can be utilized to obtain a $(2+\eps)$-approximation for scheduling preemptive jobs with non-trivial release dates on identical parallel machines.  Here, we apply our previous results to solve a {\em fast single machine relaxation}~\cite{chekuriMNS01} combined with a trick to control the actual job execution times. Then, we keep the energy assignments computed in the relaxation and apply {\em preemptive list scheduling} on parallel machines respecting release dates. We remark, that our deterministic algorithm guarantees that any solution it obtains has cost within a factor of~$2+\eps$ and it meets the energy budget. This cannot be guaranteed in the previous non-preemptive result for our objective function with energy budget in~\cite{angelBK12}. 

This paper expands considerably the extended abstract that appeared in the proceedings of \emph{ICALP '13}~\cite{MegowV13}. Among others, this new version contains complete proofs, full presentation of our techniques, and new approximation results for more general scheduling problems with release dates and identical parallel machines~(Section~\ref{sec:parallelMachines}).

\section{Model, definitions, and preliminaries} 
\label{sec:weight-space}

\subsection{Problem definition} We consider two types of scheduling problems. In both cases we are given a set of jobs $J=\{1,\ldots,n\}$ with work volumes (i.e., processing time at speed~$1$)~$v_j\geq 0$ and weights~$w_j\geq 0$. 
We seek a schedule on a single machine, described by a permutation of jobs, that minimizes the sum of weighted completion times. The speed of the machine may vary---this is where the problems distinguish.

In the problem {\em scheduling on a machine of given varying speed} we assume that the speed function~$s:\mathbb{R}_+\!\rightarrow \mathbb{R}_+$ is given implicitly by an oracle. Given a value $v$, the oracle returns the first point in time when the machine can finish $v$ units~of~work. That is, for a speed function~$s$ the oracle returns the value

\begin{equation}
\label{eq:functionF}
  f(v):= \inf\left\lbrace b>0: \int_0^b s(t) \geq v\right\rbrace.
\end{equation}

Here we are implicitly assuming that $s$ is integrable. Using the oracle, we can compute for a given order of jobs the execution and completion times and thus the total cost of the solution. We additionally must ensure that the numbers returned by the oracle can be handled efficiently. To avoid extra technical difficulties, we call an algorithm efficient if it runs in time polynomial in the input size and the largest encoding size of a number returned by the oracle.  

In the problem {\em scheduling with speed-scaling}  an algorithm determines not only a schedule
for the jobs but will also decide at which speed~$s\ge0$ the machine will run at any time. Running a machine at certain speed requires a certain amount of power. Power is typically modeled as a monomial (convex) function of speed,~$P(s)=s^{\alpha}$ with a small constant $\alpha > 1$. Given an energy budget $E$, we ask for the optimal power (and thus speed) distribution and corresponding schedule that minimizes $\sum_j w_j C_j$. More generally, we are interested in quantifying the tradeoff between the scheduling objective $\sum_{j} w_jC_j$ and the total energy consumption, that is, we aim for computing the Pareto curve for the bicriteria minimization problem. We consider two variants of speed-scaling: If the machine can run at an arbitrary speed level $s\in \mathbb{R}_+$, we say that we are in the \emph{continuous-speed} setting. On the other hand, if that machine can only choose among a finite set of speeds $\{s_1,\ldots,s_\ka\}$ we are in the \emph{discrete-speed} environment.

In both of our settings our solution concept is a permutation of jobs. Notice that this is no restriction since preemption or idle times cannot reduce the cost of the solution.

\subsection{From time-space to weight-space} For a schedule $\schedule$, we let $C_j(\schedule)$ denote the completion time of $j$ and we let $W^{\schedule}(t)$ denote the total weight of jobs completed strictly after $t$. Note that by definition $W^{\schedule}(t)$ is right-continuous, i.\,e., if $C_j(\schedule)=t$, the weight of $j$ does not count towards the remaining weight $W^{\schedule}(t)$. 
Whenever $\schedule$ is clear from the context we omit it. It is not hard to see that 
\begin{equation}
  \label{eq:obj-fct}
  \sum_{j\in J} w_jC_j(\schedule) = \int_0^\infty W^{\schedule}(t)dt.  
\end{equation}

\begin{figure}[tbp]
 
  \begin{subfigure}[b]{.49\textwidth}
    \begin{center}
      \begin{tikzpicture}[scale=0.5,yscale=.7,font=\scriptsize]
\tikzstyle{disc}=[draw,fill=white,circle,very thick,inner sep=0.4mm]
\tikzstyle{cont}=[fill=black,circle,very thick,inner sep=0.4mm]

\def\N{9}

\def\x{{0,1,2.5,4,8,10}}%
\def\y{{8.5,7,6,4.5,3,1.6,0}}%

 \draw[->] (-.1,0) -- (10.5,0) node[below right] {$t$};
 \draw[->] (0,-.1) -- (0,9.5) node[left] {$W$};

\begin{scope}[yshift=-1.2cm]
 \foreach \i/\j in {0/1 , 1/2, 2/3 , 3/4 , 4/5} {%
  \pgfmathsetmacro{\xa}{\x[\i]}
  \pgfmathsetmacro{\xb}{\x[\i+1]}
  \draw (\xa,0) rectangle node {\j} (\xb,1);  
 }
 \node at (5,-.4) {time-schedule};
\end{scope}

\begin{scope}[xshift=-1.2cm]
 \foreach \i in {0,...,3} {%
  \pgfmathsetmacro{\ya}{\y[\i]}
  \pgfmathsetmacro{\yb}{\y[\i+1]}
  \draw (0,\ya) rectangle node {\pgfmathparse{\i+1}\pgfmathprintnumber{\pgfmathresult}} (1,\yb);  
 }
  \draw[pattern=crosshatch dots] (0,\y[4]) rectangle (1,\y[5]);  
  \draw (0,\y[5]) rectangle node {5} (1,\y[6]);
  \node[rotate=90,above] at (0,5) {weight-schedule};
\end{scope}
\draw[->,thick] (-.7,2.3) -- (5,6) node[right] {idle weight};

 \foreach \i in {0,...,3} {%
  \pgfmathsetmacro{\xa}{\x[\i]}
  \pgfmathsetmacro{\xb}{\x[\i+1]}
  \pgfmathsetmacro{\ya}{\y[\i]}
  \pgfmathsetmacro{\yb}{\y[\i+1]}
  \draw[fill=gray,fill opacity=0.3] (\xa,\yb) rectangle (\xb,\ya);  
 }
 \draw[fill=gray,fill opacity=0.3] (\x[4],\y[6]) rectangle (\x[5],\y[5]);  
 \draw[pattern=crosshatch dots] (0,\y[4]) rectangle (\x[4],\y[5]);

 \foreach \i in {0,...,3} {%
  \pgfmathsetmacro{\xa}{\x[\i]}
  \pgfmathsetmacro{\xb}{\x[\i+1]}
  \pgfmathsetmacro{\ya}{\y[\i]}
  \draw[very thick]  (\xa,\ya)  -- (\xb,\ya) node[disc] {};
  }
 \draw[very thick]  (\x[4],\y[5])  -- (\x[5],\y[5]) node[disc] {};

%
%
%
%
%
%
%
%
%
%
\end{tikzpicture}
    \end{center} 
    \caption{Original solution.\newline}
    \label{fig:ganttA}
  \end{subfigure}~
  \begin{subfigure}[b]{.49\textwidth}
    \begin{center}
      \begin{tikzpicture}[scale=0.5,yscale=.7,font=\scriptsize]
\tikzstyle{disc}=[draw,fill=white,circle,very thick,inner sep=0.4mm]
\tikzstyle{cont}=[fill=black,circle,very thick,inner sep=0.4mm]

\def\N{9}

\def\x{{0,1,2,3.5,8,10}}%
\def\y{{8.5,7,6,4.5,3,2.8,1.8,1.6,0}}%

 \draw[->] (-.1,0) -- (10.5,0) node[below right] {$t$};
 \draw[->] (0,-.1) -- (0,9.5) node[left] {$W$};

\begin{scope}[yshift=-1.2cm]
 \foreach \i/\j in {0/1 , 1/3, 2/4 , 3/2 , 4/5} {%
  \pgfmathsetmacro{\xa}{\x[\i]}
  \pgfmathsetmacro{\xb}{\x[\i+1]}
  \draw (\xa,0) rectangle node {\j} (\xb,1);  
 }
 \node at (5,-.4) {time-schedule};
\end{scope}

\begin{scope}[xshift=-1.2cm]
 \foreach \i/\j in {0/1 ,1/, 2/3, 3/4 , 4/, 5/2, 6/ , 7/5} {%
  \pgfmathsetmacro{\ya}{\y[\i]}
  \pgfmathsetmacro{\yb}{\y[\i+1]}
  \draw (0,\ya) rectangle node {\j} (1,\yb);  
 }
  \draw[pattern=crosshatch dots] (0,\y[1]) rectangle (1,\y[2]);
  \draw[pattern=crosshatch dots] (0,\y[4]) rectangle (1,\y[5]);
  \draw[pattern=crosshatch dots] (0,\y[6]) rectangle (1,\y[7]);  
  \node[rotate=90,above] at (0,5) {weight-schedule};
\end{scope}

 
 \draw[fill=gray,fill opacity=0.3] (\x[0],\y[0]) rectangle (\x[1],\y[1]);
 \draw[fill=gray,fill opacity=0.3] (\x[1],\y[2]) rectangle (\x[2],\y[3]);  
 \draw[fill=gray,fill opacity=0.3] (\x[2],\y[3]) rectangle (\x[3],\y[4]);  
 \draw[fill=gray,fill opacity=0.3] (\x[3],\y[5]) rectangle (\x[4],\y[6]);  
 \draw[fill=gray,fill opacity=0.3] (\x[4],\y[8]) rectangle (\x[5],\y[7]);  
 \draw[pattern=crosshatch dots] (0,\y[1]) rectangle (\x[1],\y[2]);
 \draw[pattern=crosshatch dots] (0,\y[4]) rectangle (\x[3],\y[5]);
  \draw[pattern=crosshatch dots] (0,\y[6]) rectangle (\x[4],\y[7]);

 \foreach \i/\k in {0/0, 1/2 , 2/3 , 3/5 , 4/7 } {%
  \pgfmathsetmacro{\xa}{\x[\i]}
  \pgfmathsetmacro{\xb}{\x[\i+1]}
  \pgfmathsetmacro{\ya}{\y[\k]}
  \draw[very thick]  (\xa,\ya)  -- (\xb,\ya) node[disc] {};
  }

%
%
%
%
%
%
%
%
%
%
\end{tikzpicture}
    \end{center}
    \caption{New situation after decreasing the completion weight of Job 2.}
    \label{fig:ganttB}
  \end{subfigure}  
\caption{\textbf{2D-Gantt chart}. The $x$-axis shows a schedule, while the $y$-axis corresponds to the remaining weight function $W(\cdot)$ plus the idle weight (hatched) in the corresponding weight-schedule.
}
\label{fig:gantt}
\end{figure}
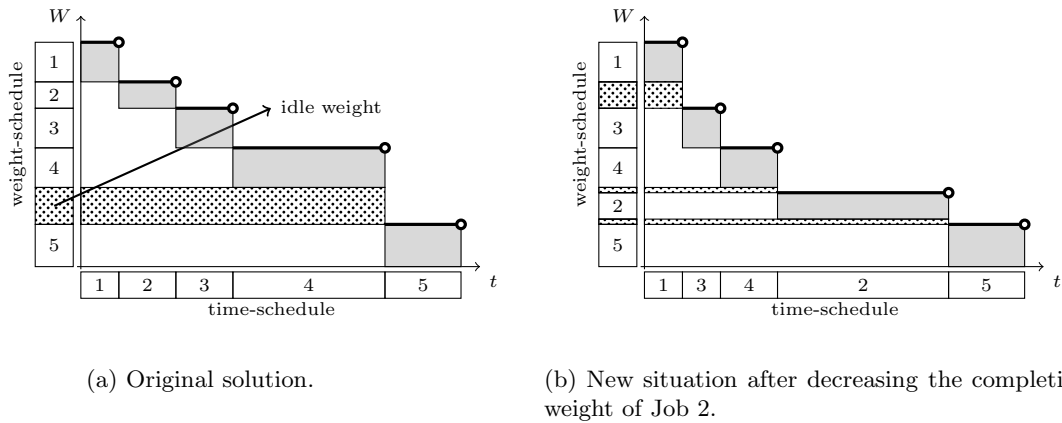

Our main idea is to describe our schedule in terms of the remaining weight function $W$. That is, instead of determining $C_j$ for each job $j$, we will implicitly describe the completion time of a job 
$j$ by the value of $W$ at the time that $j$ completes. We call this value the \emph{starting weight} of the job $j$, and denote it by $S_j^w$. Similarly, we define the \emph{completion weight} of $j$ as $C_j^w:=S_j^w+w_j$. This has a natural interpretation in the two axes of the 2D-Gantt chart (see Figure~\ref{fig:ganttA}): A typical schedule determines completion times for jobs in \emph{time-space} ($x$-axis), which is highly sensitive when the speed of the machine may vary. We call such a solution a \emph{time-schedule}. Describing a scheduling solution in terms of remaining weight can be seen as scheduling in the \emph{weight-space}~($y$-axis), yielding a \emph{weight-schedule}. 

In weight-space the weights play the role of processing times. All notions that are usually considered in schedules apply in weight-space. For example, we say that a weight-schedule is feasible if there are no two jobs overlapping, and that the machine is \text{idle} at weight value $w$ if $w\not\in [S_j^w,C_j^w]$ for all $j$. In this case we say that $w$ is \emph{idle weight} (like, for example, the hatched interval in Figure~\ref{fig:ganttA}). A non-preemptive weight-schedule immediately defines a non-preemptive time-schedule by ordering the jobs by decreasing completion weights.

Consider a weight-schedule $\schedule$ with completion weights $C_1^w\ge \ldots\ge C_n^w$, and corresponding completion times $C_{1}\le \ldots \le C_{n}$. To simplify notation let $C^w_{n+1}=0$. Then we define the cost of $\schedule$ as $\sum_{j=1}^n (C_{j}^w-C_{j+1}^w) C_j$. It is easy to check, even from the 2D-Gantt chart, that this value equals $\sum_{j=1}^n x_j^{\schedule} C_j^w$, where $x_j^{\schedule}$ is the execution time of job $j$ (in time-space). Moreover, the last expression equals Equation \eqref{eq:obj-fct} if and only if the weight-schedule does not have any idle weight. In general, the cost of the weight-schedule can only overestimate the cost of the corresponding schedule in time space, given by \eqref{eq:obj-fct}.

On a machine of varying speed, the weight-schedule has a number of technical advantages.
For instance, while creating idle {\em time} can increase the cost arbitrarily, we can create idle {\em weight} without provoking an unbounded increase in the cost. 
This gives us flexibility in weight-space and implicitly a way to delay one or more jobs in the time-schedule without increasing the cost. More precisely, we have the following observation that can be seen in the 2D-Gantt chart.
\begin{observation} 
 \label{obs:ForwardCw}
 Consider a weight-schedule $\schedule$ with enough idle weight so that
decreasing the completion weight 
 of some job $j$, while leaving the rest untouched, yields a feasible weight-schedule. This operation does
not increase the cost of the weight-schedule~$\schedule$. Indeed, notice that job $j$ is the only job for which its completion \emph{time} might increase. However, this does not increase the cost of the weight-schedule since the extra cost is dominated by the area induced by the idle weight in the original schedule.
\end{observation}

Consider Figure~\ref{fig:ganttA} as an example. Here Job 2 fits in the idle weight between Jobs~4 and~5~(hatched area). A new solution obtained by moving Job 2 to this idle weight is shown in Figure~\ref{fig:ganttB}. This operation delays Job 2 in the time-schedule, while it schedules Jobs 3 and 4 earlier. However, the total cost of the weight-schedule, i.e., the area under the curve, decreases. 

\section{A PTAS for scheduling on a machine with given speeds}
\label{sec:PTAS-given-speeds}

In what follows we give a PTAS for minimizing $\sum_j w_jC_j$ on a machine with a given speed function. In order to gain structure, we start by applying several modifications to the instance and optimal solution. Consider $0 <\e < 1/2 $. First we round the weights of the jobs to the next integer power of $1+\e$, which increases the objective function by at most a factor $1+\e$.

Additionally, we discretize the weight-space in intervals that increase exponentially. That is, we consider intervals $I_u=[(1+\e)^{u-1},(1+\e)^u)$ for $u\in \{1,\ldots,\nu\}$ where $\nu:= \lceil \log_{1+\e}\sum_{j\in J} w_j\rceil$. We denote the length of each interval $I_u$ as $|I_u|:=\e (1+\e)^{u-1}$. 

We will apply two important procedures to modify weight-schedules. They are used to create idle weight so to apply Observation~\ref{obs:ForwardCw}, and they only increase the total cost by a factor $1+\OO(\e)$. Similar techniques, applied in time-space, were used by Afrati et al.~\cite{afrati_approximation_1999} for problems on constant-speed machines.

\begin{enumerate}
\item \textit{Weight Stretch:} We multiply by $1+\e$ the completion weight of each job. This creates an idle weight interval of length $\e w_j$ before the starting weight of job $j$. This operation increases the cost by a $1+\e$ factor.
\item \textit{Stretch Intervals:} 
We delay the completion weight of each job $j$ with $C_j^w\in I_u$ by $|I_u|$, so that $C_j^w$ belongs to $I_{u+1}$. 
Then 
$|I_{u+1}|-|I_{u}|=\e^2(1+\e)^{u-1}=\e|I_{u+1}|/(1+\e)$ units of weight are left idle in $I_{u+1}$ after the transformation, unless there was only one job completely covering $I_{u}$. By moving jobs within $I_{u+1}$, we can assume that this idle weight is consecutive.
 This transformation increases the cost by at most a factor $(1+\e)^2=1+\OO(\e)$. 
\end{enumerate}

\subsection{Dynamic program} 
\label{sec:DPPTAS}

We now show our dynamic programming (DP) approach to obtain a PTAS. We first describe a DP table with exponentially many entries and then discuss how to reduce its size. Recall that our schedules in time-space do not use idle time. Therefore we can uniquely describe a schedule by specifying a non-preemptive weight-schedule and ordering the jobs accordingly in the time-axis.


Consider a subset of jobs $S\subseteq J$ and a partial schedule of $S$ in the weight-space. In our dynamic program, $S$ will correspond to the set of jobs at the beginning of the weight-schedule, i.\,e., if $j\in S$ and $k\in J\setminus S$ then $C_j^w< C_k^w$. A partial weight-schedule $\schedule$ of jobs in $S$ implies a schedule in time-space with the following interpretation. Note that the makespan of the time-schedule is completely defined by the total work volume $\sum_j v_j$. We impose that the last job of the schedule, which corresponds to the first job in $\schedule$, finishes at the makespan. This uniquely determines a value of $C_j$ for each $j\in S$, and thus also its execution time $x_j^{\schedule}$. The total cost of this partial schedule is $\sum_{j\in S} x_j^{\schedule}C_j^w$.

 Consider $\mathcal{F}_u:=\{S\subseteq J: w(S)\le (1+\e)^u\}$. That is, a set $S\in \mathcal{F}_u$ is a potential set of jobs whose completion weight belongs to $I_{u'}$ with $u'\le u$. For a given interval $I_u$ and set $S\in \mathcal{F}_u$, we construct a table entry $T(u,S)$ with a $(1+\OO(\e))$-approximation to the optimal cost of a weight-schedule of $S$ subject to $C_j^w \le (1+\e)^{u}$ for all $j\in S$.
  
 Consider now $S\in \mathcal{F}_u$ and $S'\in \mathcal{F}_{u-1}$ with $S'\subseteq S$. Let $\schedule$ be a partial schedule of $S$ where the set of jobs with completion weight in $I_u$ is exactly $S\setminus S'$. We define $\APPROX_u(S',S)=(1+\e)^u\sum_{j\in S\setminus S'} x_j^{\schedule}$, which is a $(1+\e)$-approximation to  $\sum_{j\in S\setminus S'} x_j^{\schedule}C_j^w$, the partial cost associated to $S\setminus S'$. We remark that the values $\sum_{j\in S\setminus S'} x_j^{\schedule}$ and $\APPROX_u(S',S)$ do not depend on the whole schedule $\schedule$, but only on the total work volume of jobs in $S'$. 

 We can compute $T(u,S)$ with the following formula, 
 \[
   T(u,S) = \min\{ T(u-1,S') + \APPROX_u(S',S) :S'\in \mathcal{F}_{u-1}, S'\subseteq S\}.  
 \]

The set $\mathcal{F}_u$ can be of exponential size, and thus also this DP table. In the following we show that there is a polynomial size set $\tilde{\mathcal{F}}_u$ that yields $(1+\e)$-approximate solutions. We remark that the set $\tilde{\mathcal{F}}_u$ will not depend on the speed of the machine. Thus, the same set can be used in the speed-scaling scenario.

\subsection{Light jobs}
\label{sec:light}

We structure an instance by classifying jobs by their size in weight-space. This classification allows us to determine the schedule of part of the jobs greedily, which will help to define $\tilde{\mathcal{F}}_u$ properly.

\begin{definition}
  Given a schedule and a job $j$ with starting weight $S_j^w\in I_u$, we say that $j$ is \emph{light for $S_j^w$} if $w_j\le \e^2 |I_u|$. A job that is not light is \emph{heavy for $S_j^w$}.
\end{definition}

To simplify notation, we say that a job is light or heavy when the starting weight $S_j^w$ is clear from the context. 

Given a weight-schedule for heavy jobs, we give a greedy algorithm to schedule light jobs that increases the cost by a $1+\OO(\e)$ factor. Consider any weight-schedule $\schedule$. First, remove all light jobs. Then we move jobs within each interval $I_u$, such that the idle weight in $I_u$ is consecutive. Clearly, this can only increase the cost of the solution by a factor $1+\e$. Then, we apply a preemptive greedy algorithm to assign light jobs, namely, Smith's rule~\cite{Smith56}. More precisely, for each idle weight $w$ we process the job $j$ that maximizes $v_j/w_j$ among jobs that are not completely processed yet and $w_j\le \e^2 |I_u|$. (Here we give priority to jobs with smallest weight to work volume ratio, which is the opposite as to normal Smith's rule; intuitively, this is because in weight-space jobs are scheduled in reversed order as in time-space.) To remove preemptions, we apply the Stretch Interval subroutine\footnote{The Stretch Interval procedure also applies to preemptive settings by interpreting each piece of a job as an independent job.}, creating an idle weight interval in $I_u$ of length at least $\e|I_u|/(1+\e)\ge \e|I_u|/2\ge \e^2 |I_u|$ (since $\e\le1/2$). This gives enough space in each interval $I_u$ to completely process the (unique) preempted light job with starting weight in $I_u$. The algorithm returns this last schedule, called~$\schedule'$. 
Summarizing, the algorithm is as follows.

\begin{algorithmNB}{Smith in Weight-Space}
\noindent \textit{Input:} A weight-schedule $\schedule$.

 \begin{enumerate}
  \item Remove all light jobs in $\schedule$ and move the remaining jobs within each interval $I_u$, such that the idle weight in $I_u$ is consecutive. 
  \item \emph{Reverse Smith's rule}: For $u=1,\ldots,\nu$ and each idle weight $w\in I_u$, process at $w$ a job $j$ maximizing $v_{j}/w_{j}$ among all available jobs with $w_j\le \e^2 |I_u|$.
  \item Apply the Stretch Intervals subroutine.
  \item For each $u$ move the unique preempted light job  with starting weight in $I_u$ (if any) so that it is completely processed within $I_u$.
  \item Return the constructed schedule $\schedule'$.
 \end{enumerate}
\end{algorithmNB}

 We now show that the cost of the schedule $\schedule'$ returned by the algorithm is at most a factor of $1+\OO(\e)$ larger than the cost of $\schedule$. To do so we need a few definitions.

\begin{definition} Given a weight-schedule $\schedule$, its \emph{remaining volume function} is defined as 
\[V^{\schedule}(w) := \sum_{j: C_j^w \ge w} v_j.\]\jcom{Check whether $C_j^w>w$} 
\end{definition}

Consider now the function $f(v)$ corresponding to the earliest time by which the machine can have processed a work volume of $v$, i.e., the function defined in Equation~\eqref{eq:functionF}. 
 It is easy to see---even from the 2D-Gantt chart---that $\int_{0}^{\infty} f(V^{\schedule}(w))dw$ corresponds to the cost of the weight-schedule $\schedule$. Also, notice that $f(v)$ is non-decreasing, so that $V^{\schedule}(w)\le V^{\schedule'}(w)$ for all $w\ge0$ implies that the cost of $\schedule$ is at most the cost~of~$\schedule'$. 
 
\begin{definition} For a given $w$, let $I_j(w)$ be equal $1$ if the weight-schedule processes $j$ at weight $w$, and $0$ otherwise. Then, $\chi_j(w):=(1/w_j)\int_w^{\infty} I_j(w')dw'$ corresponds to the fraction of job $j$ processed after $w$. The \emph{fractional remaining volume function} of a weight-schedule $\schedule$ is defined as
\[
  V^{\schedule}_f(w) := \sum_{j: j \text{ is light}} \chi_j(w)\cdot v_j + \sum_{j:j \text{ is heavy},C_j^w \ge w } v_j \qquad \text{for all }w\ge0.
\]
\end{definition}
Intuitively, this function is similar to the (non-fractional) remaining volume function with the difference that it treats light jobs as ``liquid''. Also, notice that $V^{\schedule}_f(w)\le V^{\schedule}(w)$ for all~$w\ge0$. 

\begin{lemma}
 \label{lm:AlgSmith}
 Let $\schedule$ be a weight-schedule and $\schedule'$ be the output of Algorithm Smith in Weight-Space on input $\schedule$. Then the cost of $\schedule'$ is at most a factor $1+\OO(\e)$ larger than the cost of $\schedule$.
\end{lemma}
\begin{proof}
Let $\schedule_i$ be the schedule constructed after Step $i$ of the algorithm for each $i\in \{1,2,3,4\}$. In particular, $\schedule_1$ schedules only heavy jobs and $\schedule_4=\schedule'$. First we observe that for any given $w\ge0$, $V^{\schedule_2}_f(w)$ is a lower-bound on $V^{\hat{\schedule}}_f(w)$ for any schedule $\hat{\schedule}$ that coincides with $\schedule_2$ on the heavy jobs. This follows by a simple exchange argument, since the greedy Smith-type rule in Step 2 chooses the job that packs as much volume as possible in the available weight among all light jobs. We conclude that $V_f^{\schedule_2}(w)\le V_f^{\schedule}(w)$ for all $w$.

Observe that applying Stretch Intervals can delay any piece of a job by at most a factor $(1+\e)^2$. Therefore $V_f^{\schedule_3}(w)\le V_f^{\schedule_2}((1+\e)^{-2}w)$. Also, in Step 4 pieces of jobs are only moved backwards and thus $V_f^{\schedule_4}\le V_f^{\schedule_3}$. Finally, we notice that each light jobs in $\schedule_4$ is processed completely within an interval $I_u$, and thus $V^{\schedule_4}(w)\le V_f^{\schedule_4}((1+\e)^{-1}w)$. 

Combining all of our observations we obtain that 
\[
 V^{\schedule_4}((1+\e)^3w)\le V_f^{\schedule_4}((1+\e)^{2}w)\le V_f^{\schedule_3}((1+\e)^{2}w)\le V_f^{\schedule_2}(w)\le V_f^{\schedule}(w) \le V^{\schedule}(w)\quad \text{for all }w\ge0.
\]
Taking the function $f(\cdot)$ and integrating implies that
\[
 \int_{0}^{\infty} f(V^{\schedule_4}((1+\e)^{3}w))dw\le \int_{0}^{\infty} f(V^{\schedule}(w))dw.
\]
 Finally, the right hand side of this inequality is the cost of $\schedule$, and a simple change of variables implies that the left hand side is $(1+\e)^{-3}$ times the cost of $\schedule'=\schedule_4$. The lemma follows.
\end{proof}

The next corollary follows directly from our previous result. 
\begin{corollary}
  \label{cor:SmithOPT}
  At a loss of a factor $1+\OO(\e)$ in the objective function, we can assume the following. For a given interval $I_u$, consider any pair of jobs $j,k$ whose weights are at most $\e^2|I_u|$. If both jobs are processed in $I_u$ or later and $v_k/w_k\le v_{j}/w_j$, then $C_j^w\le C_k^w$.
\end{corollary}

\subsection{Localization}
\label{sec:localization}

The objective of this section is to compute, for each job $j\in J$, two values $r_j^w$ and $d_j^w$ so that job $j$ is scheduled completely within $[r_j^w,d_j^w)$ in some $(1+\OO(\e))$-approximate weight-schedule.  We call $r_j^w$ and $d_j^w$ the \emph{release-weight} and \emph{deadline-weight} of job $j$, respectively. Crucially, we need that the length of the interval $[r_j^w,d_j^w)$ is not too large, namely that $d_j\in \OO(\text{poly}(1/\e) r_j)$. Such values can be obtained by using Corollary~\ref{cor:SmithOPT} and techniques from~\cite{afrati_approximation_1999}. The release- and deadline-weights will help us finding a compact set $\tilde{\mathcal{F}}_u$.

We consider an initial value for $r_j^w$ and then increase its value iteratively. We will restrict ourselves to values of $r_j^w$ that are integer powers of $1+\e$. Consider an arbitrary weight-schedule. Recall that for a job with completion weight $C_j^w$, the Weight Stretch subroutine increases the completion weight $(1+\e)C_j^w$ and hence the starting weight to $S_j^w = \e C_j^w$. Applying the procedure twice we get a solution that satisfies $S_j^w\ge \e (1+\e)C_j^w\ge \e (1+\e)w_j$. Thus, we can safely define $r_j^w$ as $\e w_j$ rounded up to an integer power of $1+\e$.

We now show how to adapt techniques from~\cite{afrati_approximation_1999} used for time-schedules. Let $J_u$ be the set of all jobs with $r_j^w$ equal to $(1+\e)^{u-1}$. We partition $J_u$ into light and heavy jobs, depending if their weight is smaller or larger than $\e^2 |I_u|$. 
Note that a heavy job in $J_u$ can have weights $w$ with $\e^2 |I_u| < w \le 1/\e (1+\e)^{u-1}$, where the last inequality follows since $r_j^w\ge \e w_j$. Therefore, since we are assuming that the weights of jobs are integer powers of $1+\e$, for a fixed $u$ we only need to consider heavy jobs with weights

\[w\in \Omega_u:=\left\lbrace(1+\e)^i: \e^2 |I_u| < (1+\e)^i\le \frac{(1+\e)^{u-1}}{\e}, \text{ where }i\in \mathbb{Z} \right\rbrace.\] 

Crucially, note that  $|\Omega_u|\in \OO(\log_{1+\e} 1/\e)\subseteq \OO(1/\e\cdot \log 1/\e )$. Based on this we give the following decomposition of the set of jobs with a given release-weight.
\begin{definition}
 Given release-weights for each job, we define $J_u=\{j: r_j^w=(1+\e)^{u-1}\}$. Additionally, we decompose $J_u$ into a set of light jobs $L_u:=\{j \in J_u: w_j\le \e^2 |I_u|\}$, and sets $H_{u,w}=\{j\in J_u\,:\,w_j=w\}$ of heavy jobs of weight $w$ for each $w\in \Omega_u$.
\end{definition}

Now we consider all jobs in $L_u$. If $w(L_u)$ is larger than $(1+\e^2)|I_u|$ then some jobs in $L_u$ will have to start in $I_{u+1}$ or later. By Corollary~\ref{cor:SmithOPT} we can choose the set of possible jobs with starting weight in $I_u$ greedily, and increase the release-weight of the rest. Similarly, since the weight of each job in $H_{u,w}$ is the same, we can always give priority to jobs with the largest work volume. With this idea we can show the following lemma.

\begin{lemma}
\label{lm:boundRelease}
 We can compute in polynomial time release-weights $r_j^w$ for each job $j$ such that there exists a $(1+\OO(\e))$-approximate weight-schedule respecting the release-weights and for any interval~$I_u$ we have that $w(J_u) \in \OO(1/\e^3\cdot \log 1/\e \cdot|I_u|).$ And this weight-schedule satisfies the property of Corollary~\ref{cor:SmithOPT}.
\end{lemma}
\begin{proof}
 Initialize $r_j^w$ as $\e w_j$ rounded up to an integer power of $(1+\e)$ and let $J_u, L_u$ and $H_{u,w}$ be defined as above. By Corollary~\ref{cor:SmithOPT} we know that within an interval $I_u$ we can order light jobs and process first the job with largest $v_j/w_j$ ratio. Thus, if the total weight of jobs in $L_u$ is larger than $(1+\e^2)|I_u|$ we increase the release-weight of a job $j^*\in \arg\min_{j\in L_u} v_j/w_j$ to $(1+\e)^{u}$. Note that after doing this $j^*$ does not belong to $L_u$ anymore. We iterate this procedure until $w(L_u)\le (1+\e^2)|I_u|$.

 We do a similar technique for jobs in $H_{u,w}$. If $w(H_{u,w})> |I_u|+w$ and $|H_{u,w}|$ contains more than one job, then we can delay the release-weight of a job $j^*\in H_{u,w}$ with smallest $v_j$. This follows by a simple interchange argument, since if there are two jobs with the same weight then the one with smallest work has smaller (larger) completion time (weight) in an optimal solution. After modifying the release date of $j^*$ this job does not belong to $H_{u,w}$ anymore.

 This way we obtain a set $H_{u,w}$ with
 \[
  w(H_{u,w})\le |I_u| + w\le |I_u| + \frac{1}{\e} (1+\e)^{u-1} \in \OO(1/\e^2)\cdot |I_u|.
 \]
 We execute the two procedures described above for each $u=0,\ldots,\nu$ where $\nu=\lceil \log_{1+\e} \sum_{j\in J} w_j \rceil$ until the following property holds: for all $u\in \{0,\ldots,\nu\}$ and $w\in \Omega_u$ we have that $w(L_u)\le (1+\e^2)|I_u|$ and $w(H_{u,w})\in \OO(1/\e^2)\cdot |I_u|$. The result follows since $|\Omega_u|\in \OO(1/\e \cdot \log 1/\e )$. 
 \end{proof}

 We use the previous lemma to define the deadline-weights by using the following idea. 
 For~$s$ large enough (but constant), Stretch Intervals creates enough idle weight in $I_{u+s}$ to fit all jobs released at $(1+\e)^u$ that have not yet finished by $(1+\e)^{u+s+1}$. This allows us to apply Observation~\ref{obs:ForwardCw}. 

\begin{lemma} 
\label{lm:localization}
We can compute in polynomial time values $r_j^w$ and $d_j^w$ for each $j\in J$ such that: (i) there exists a $(1+\OO(\e))$-approximate weight-schedule that processes each job $j$ within $[r_j^w,d_j^w)$, (ii) there exists a constant $s\in \OO(\log(1/\e)/\e)$ such that $d_j^w\le r_j^w\cdot (1+\e)^s$, (iii) $r_j^w$ and $d_j^w$ are integer powers of $(1+\e)$, (iv) within each $L_u$ jobs are processed following Reverse Smith's rule, and (v) the values $r_j^w$ an $d_j^w$ are independent of the speed of the machine.
\end{lemma}
\begin{proof}
Consider the release-weights given by the previous lemma and consider the associated sets $J_u$ for each $u$. Then, since $w(J_u) \in \OO(1/\e^3\cdot \log 1/\e \cdot|I_u|)$, there exists an integer $s\in \OO(\log_{1+\e} (1/\e^4\cdot \log 1/\e))\subseteq \OO(\log(1/\e)/\e)$ such that $w(J_u)\le \e |I_{u+s-1}|/(1+\e)$.

 Consider now the $(1+\OO(\e))$-approximate solution obtained from the previous lemma (which, by construction, also satisfies the property of Corollary~\ref{cor:SmithOPT}). By construction of $r_j^w$, we can assume that the starting weight of $j$ in this schedule is at least $r_j^w$. Now we apply Stretch Intervals. This creates $\e |I_{u+s-1}|/(1+\e)$ idle weight in interval $I_{u+s-1}$, unless there was one job completely covering $I_{u+s-1}$. If that is not the case, then we can move all jobs in $J_u$ with starting weight in $I_{u+s}$ or larger to be completely processed inside $I_{u+s-1}$. By Observation~\ref{obs:ForwardCw}, doing this can only increase the objective function by a $1+\OO(\e)$ factor. Similarly, if there was a job $k$ completely covering $I_{u+s-1}$, then the idle weight that $I_{u+s-1}$ should have contained can be considered to be just before the starting weight of $k$. In this case we can move all jobs in $J_u$ that were being processed after $I_{u+s-1}$ to just before~$S_k^w$. 

 In either case we constructed a solution where each job in $J_u$ is completely processed in $[(1+\e)^{u-1}, (1+\e)^{u+s-1})$. Properties (i)-(iii) in the lemma follows by defining $d_j^w = (1+\e)^{u+s-1}=r_j^w (1+\e)^s$ for each job ~$j\in J_u$. Also property (iv) follows since our original schedule satisfies the property of Corollary~\ref{cor:SmithOPT} and our modification does not change the relative order of jobs in $J_u$. Finally (v) follows since while defining $r_j^w$ and $d_j^w$ we never used the speed of the machine.
\end{proof}

\subsection{Compact Search Space}
\label{sec:search}


Given the job classification and localization in the previous subsections, we are now ready to reduce the running time of the dynamic program in Section~\ref{sec:DPPTAS} to polynomial time. To that end, recall the definition of families of job sets $\mathcal{F}_u$. 
We will define a polynomial-size version of it,~$\tilde{\mathcal{F}}_u$.
Instead of describing a set $S\in \tilde{\mathcal{F}}_u$, we describe $R=J\setminus S$, that is, the jobs with completion weights in $I_{u+1}$ or later. 
 That is, we define a set $\mathcal{D}_u$ that will contain the complements of sets in $\tilde{\mathcal{F}}_u$.
 In order to define $\mathcal{D}_u$ we use the release- and deadline-weights given by Lemma~\ref{lm:localization}. If $R\in \mathcal{D}_u$, then $R$ must contain all jobs $j\in \overline{R}:=\{k\in J: r_k^w\ge (1+\e)^u\}$. 
\begin{observation}
\label{obs:Rdecomp0}
Each set $R\in \mathcal{D}_u$ is of the form $R'\cup \overline{R}$, where every job $j\in R'$ has $r_j^w\le (1+\e)^{u-1}$. 
\end{observation}

Thus we only need to describe all possibilities for $R'$. For a job $j\in R'$ we can assume that $d_j^w \ge (1+\e)^{u+1}$. Therefore, by Lemma~\ref{lm:localization}, we have that $r_j^w\ge (1+\e)^{u+1-s}$, where $s\in \OO(\log(1/\e)/\e)$. 

\begin{observation}
\label{obs:Rdecomp}
 Each set $R=R'\cup \overline{R}\in \mathcal{D}_u$ is of the form $ \left(\displaystyle{ \bigcup_{v=u+2-s}^{u} R'_v}\right)\cup \overline{R}$, where $R'_v:=\{j\in R': r_j^w=(1+\e)^{v-1}\}$.
\end{observation}

Then, we aim to find a collection of subsets that can play the role of $R'_v$. If the size of this collection is at most a polynomial number $k$, we could conclude that $|\mathcal{D}_u|\le k^{s-1}= k^{\OO(\log(1/\e)/\e)}$.

In order to do so, recall that $J_v$ denotes the set of all jobs with release-weights equal to $(1+\e)^{v-1}$, and that we can write $J_v=L_v \cup (\bigcup_w H_{v,w})$ where $w\in\Omega_v$ and $|\Omega_v|\in \OO(\log_{1+\e}1/\e)$. Thus, defining $R'_{v,w}:= R'_v\cap H_{v,w}$ we can further decompose $R'_v$ as $(R'_v\cap L_v)\cup(\bigcup_w R'_{v,w})$. Now notice that $R'_{v,w}$ is a subset of $H_{v,w}$ which, as seen in the proof of the next observation, has a very simple structure.

\begin{observation}
\label{obs:NumberHeavy}
 Without loss of generality, we can restrict ourselves to consider sets $R'_{v,w}$ among $\OO(1/\e^2)$ distinct options.
\end{observation}
\begin{proof}
 Let $w\in \Omega_v$. Each job in $H_{v,w}$ has weight $w$ and, as seen in the proof of Lemma~\ref{lm:boundRelease}, we have that $w(H_{v,w}) \le |I_v|+w$. Thus $H_{v,w}$ contains at most $1+|I_v|/w$ many jobs. Since by definition of $H_{v,w}$ we have that $w\ge \e^2 |I_v|$, we obtain that $|H_{v,w}|\in  \OO(1/\e^2)$. Moreover, all jobs in $H_{v,w}$ has the same weight $w$ and the same release-weight. Therefore, we know that these jobs are ordered by their work volume in an optimal solution. Thus, we can restrict ourselves to sets $R'_{v,w}$ that respect this order. The observation follows since there are at most $|H_{v,w}|+1\in \OO(1/\e^2)$ many sets that respect this order. 
\end{proof}

Given $v$, the index $w$ ranges over $|\Omega_v|\in\OO(\log(1/\e)/\e)$ many values. Thus the following holds.

\begin{observation}
 \label{obs:NumberHeavyAll}
 For each $v$ the set $\bigcup_{w} R'_{v,w}$ can be chosen over $(1/\e^2)^{\OO(\log(1/\e)/\e)} =2^{\OO(\log(1/\e)^2/\e)}$ many alternatives.
\end{observation}
 
We use a similar argument for $R'_v\cap L_v$. Indeed, as seen in the proof of Lemma~\ref{lm:boundRelease}, $w(L_v)\le (1+\e^2)|I_v|$ and jobs in $L_v$ will be processed as light jobs (by Lemma~\ref{lm:localization}). We now show that we can group light jobs together in order to diminish the possibilities for $L_v$. This is done as follows. Set jobs in $L_v$ in a list ordered by Reverse Smith's rule, as in Algorithm Smith in Weight-Space. Then we greedily find groups of jobs in $L_v$ by going through the list of jobs from left to right such that each group has total weight in $[\e^2|I_v|,2\e^2|I_v|]$ (except from the last group that might have smaller total weight). Recalling that $w(L_v)\in (1+\e^2)|I_v|$, we obtain that this procedure creates at most $\OO(1/\e^2)$ groups. Let $L_{v,i}$ be the $i$-th of these groups.

\begin{lemma}
 There exists a $(1+\OO(\e))$-approximate weight-schedule such that: (i) it satisfies the release- and deadline-weights of Lemma~\ref{lm:localization}, (ii) in each group $L_{v,i}$ all jobs are processed consecutively, and (iii) within each set $L_v$ jobs are processed following Reverse Smith's rule. 
\end{lemma}
\begin{proof}
 Consider the schedule given by Lemma~\ref{lm:localization}, and thus within each $J_v$ jobs follow Reverse Smith's rule. Let us fix an interval $I_{v'}$. Within this interval, the schedule can only process jobs in $J_v$ with $v\leq v'$. Within a given $J_v$ we follow Reverse Smith's rule, thus there is at most two sets $L_{v,i}$ that are partially processed in $I_{v'}$. They require at most $4\e^2 |I_v|$ extra weight within $I_{v'}$ in order to be completely processed in $I_{v'}$. Summing over all $v\le v'$, we obtain that in total we require 
 \[
 4\e^2 \sum_{v\le v'} |I_v| = 4\e^3 \sum_{v\le v'} (1+\e)^v \in \OO(\e |I_{v'}|)  
 \]
 extra space in $I_{v'}$. The result follows since we can create enough idle time within $I_{v'}$ by applying $\OO(1)$ 
times the procedure Stretch Intervals. We remark that the procedure described works \emph{simultaneously} for all intervals $I_{v'}$. 
\end{proof}

With this lemma, we can find a compact description to $R'_v\cap L_v$. Indeed, to specify $R'_v\cap L_v$, i.\,e., the jobs in $L_v$ that are processed in $I_{u+1}$ or later, we just need to determine the index $i$ such that jobs in $L_{v,k}$ with $k\ge i$ are in $R'_v$ and jobs in $L_{v,k}$ with $k<i$ are not in $R'_v$. Since $i$ ranges over $\OO(1/\e^2)$ many options, we obtain the following.

\begin{observation}
 The set $R'_v\cap L_v$ can be chosen over $\OO(1/\e^2)$ different options.
\end{observation}


Combining this last observation and Observation~\ref{obs:NumberHeavyAll}, we obtain that $R'_v$ can take at most $k\le 2^{\OO(\log^2(1/\e)/\e)}$ many different options. By Observation~\ref{obs:Rdecomp}, we conclude that $R'$ belongs to a set of size at most $k^{s-1}\le 2^{\OO(\log^3(1/\e)/\e^2)}$. With this and Observation~\ref{obs:Rdecomp0}, we can define $\mathcal{D}_u$ having size at most $2^{\OO(\log^3(1/\e)/\e^2)}$. Finally, we define $\tilde{\mathcal{F}}_u=\{R: R^c\in \mathcal{D}_u\}$ for each $u$.

\begin{lemma}
 \label{lm:compactFu}
 We can construct in polynomial time a set $\tilde{\mathcal{F}}_u$ for each $u$ that satisfies the following: (i) there exists a $(1+\OO(\e))$-approximate weight-schedule in which the set of jobs with completion weight at most $(1+\e)^u$ belongs to $\tilde{\mathcal{F}}_u$ for each interval $u$, (ii) the set $\tilde{\mathcal{F}}_u$ has cardinality at most $2^{\OO(\log^3(1/\e)/\e^2)}$, and (iii) the set $\tilde{\mathcal{F}}_u$ is completely independent of the speed of the machine.
\end{lemma}

With this lemma and the discussion at the beginning of this section we obtain a PTAS, which is 
 best possible from an approximation point of view, since the problem is known to be 
strongly~NP-hard~\cite{hoehnJ12}.

\begin{theorem}\label{thm:PTAS-given-speeds}
  There exists an efficient PTAS 
  for minimizing the weighted sum of completion times on a machine with given varying speed. 
\end{theorem}
\begin{proof}
It remains to argue that the described algorithm is efficient. It is easy to see that the time for creating sets $\tilde{\mathcal{F}}_u$ is dominated by the time needed to solve the dynamic program. Moreover, the number of entries of the table is $2^{\OO(\log^3(1/\e)/\e^2)}\cdot \log(\sum_j w_j)$, and the time needed to fill each entry is $2^{\OO(\log^3(1/\e)/\e^2)}\cdot n$. Therefore the running time\footnote{We remark that in this expression we consider arithmetic operations to take time $\OO(1)$, and thus we neglect the size of the numbers output by the oracle. However considering this effect can only add a polynomial term on the maximum encoding size of a number output by the oracle. Recall that we allow efficient algorithms to be of that form.} is $2^{\OO(\log^3(1/\e)/\e^2)}\cdot \log(\sum_j w_j)\cdot 2^{\OO(\log^3(1/\e)/\e^2)}\cdot n = 2^{\OO(\log^3(1/\e)/\e^2)}\cdot \log(\sum_j w_j)\cdot n$.
\end{proof}

\section{Speed-scaling for continuous speeds}
\label{sec:cont-speeds}

We now consider the dynamic speed-scaling setting in which the machine can run at any non-negative speed $s$, and it is part of the scheduling problem to decide upon the speed. Running the machine at speed~$s$ implies a power consumption rate of $P(s)=s^{\alpha}$ for some constant $\alpha\ge 1$. The total energy consumed is the power consumption integrated over time. We study the problem of minimizing $\sum_j w_j C_j$ for a given amount of available energy $E$.

In this setting, we may assume that an optimal solution executes each job at a uniform speed. This follows directly from the convexity of the power function~\cite{yaoDS95}. Let~$s_j$ be the speed at which job~$j$ is
running. Then~$j$'s power consumption is~$p_j=s_j^{\alpha}$, and its
execution time
is~$x_j=v_j/s_j=v_j/p_j^{1/\alpha}$. The energy that
is required for processing~$j$ is~$E_j=p_j\cdot x_j = p_j\cdot v_j/ s_j = s_j^{\alpha -1}\cdot v_j=
v_j^{\alpha}/x_j^{\alpha-1}$. 

Let $\pi$ be a sequence of jobs in a schedule, where $\pi(j)$ is the index of
the $j$-th job in the sequence. We can compute the
optimal energy assignment for all jobs in a given sequence~$\pi$ using a total
amount of energy~$E$ by a convex program. We rewrite the
objective
function as $\sum_{j=1}^n w_j C_j = \sum_{j=1}^n w_{\pi(j)}
\sum_{k=1}^j x_{\pi(k)} = \sum_{j=1}^n x_{\pi(j)} \sum_{k=j}^n w_{\pi(k)}$ and
define
$W^{\pi}_{\pi(j)}=\sum_{k=j}^n w_{\pi(k)}$. Note that
$x_j=\left(v_j^{\alpha}/E_j\right)^{1/(\alpha-1)}$, and that $W^{\pi}_j$ is the
total remaining weight just before $j$ is completed in any schedule concordant
with~$\pi$. Then the problem can be formulated as
\begin{equation}
\label{eq:SWCTcontSpeed}
\min \left\lbrace \sum_{j=1}^n W_j^{\pi} \cdot 
\left(\frac{v_j^{\alpha}}{E_j}\right)^{1/(\alpha-1)}: \quad \sum_{j=1}^n E_j
\leq E \text{, and }   E_j \geq 0 \quad \forall j \in                      
\{1,\ldots,n\}\right\rbrace. 
\end{equation}

This program has linear constraints and a convex objective
function. Such programs can be solved in polynomial time up
to an arbitrary precision~\cite{nesterovN94} with the Ellipsoid method. 
However, the well-known Karush-Kuhn-Tucker (KKT)~\cite{bertsekas99} conditions yield a explicitly formula for the optimal energy assignment.

The problem in~\eqref{eq:SWCTcontSpeed} is clearly feasible, for example, choose $E_j=0$ for each~$j\in\{1,\ldots,n\}$. Moreover, an optimal solution satisfies the first constraint with equality. Indeed, we
allow arbitrary non-negative speeds and thus arbitrary energy assignments, and the
smallest increase in the assigned energy decreases the total
cost. For the same reason and with a positive energy budget, an
optimal solution never assigns zero energy to any job; hence $E_j> 0$ for each job $j$. With these observations the KKT conditions reduce to the following.

\begin{lemma}[KKT conditions] 
  \label{lem:kkt} A vector $(E_1,\ldots,E_n)$ is an optimal solution
to the convex program in \eqref{eq:SWCTcontSpeed} if and only if
  \begin{itemize}
  \item[(a)] $(E_1,\ldots,E_n)$ is feasible and satisfies $\sum_{j=1}^n E_j
= E$ and  $E_j > 0$ for all $j$, and
  \item[(b)] there exists a parameter $\lambda \geq 0$ such that
$\nabla g(E_1,\ldots,E_n) + \lambda \cdot \boldsymbol{1} = 0$,
  \end{itemize}
 where $\boldsymbol{1}$ denotes a vector with ones in each coordinate and $g$ is the objective function in \eqref{eq:SWCTcontSpeed}.
\end{lemma}

\begin{theorem}\label{thm:E_j} 
The optimal solution to \eqref{eq:SWCTcontSpeed} is given by

\begin{equation*}
E_j = v_j \cdot
\left(W_j^{\pi}\right)^{(\alpha-1)/\alpha} \cdot \frac{E}{\gamma_{\pi}} \text{,
where } \gamma_{\pi} = \sum_{j=1}^n v_{j} \cdot
\left(W_{j}^{\pi}\right)^{(\alpha-1)/\alpha}.
  \end{equation*}
\end{theorem}
\vspace{-.3cm}
\begin{proof}
  Since we fix a permutation $\pi$, we omit the
 extra script in $W_j^{\pi}$ and $\gamma_{\pi}$ during the rest of this proof. Let $(E_1,\ldots,E_n)$ be an
optimal solution to~\eqref{eq:SWCTcontSpeed}. By Lemma~\ref{lem:kkt}(b),
  there is a~$\lambda \geq 0$ such that for every job~$j\in
  \{1,\ldots,n\}$ holds
  \begin{equation*} 
    W_j \cdot v_j^{\alpha/(\alpha-1)} \cdot
    \frac{-1}{\alpha-1} \cdot E_j^{-\alpha/(\alpha-1)} + \lambda = 0\,,
  \end{equation*} which is equivalent to
  \begin{equation}\label{eq:E_j} E_j = v_j \cdot
    W_j^{(\alpha-1)/\alpha} \cdot
    \left(\frac{1}{(\alpha-1)\lambda}\right)^{(\alpha-1)/\alpha} \,.
  \end{equation} 
  To determine the Lagrange multiplier~$\lambda$ we use Lemma~\ref{lem:kkt}(a),
  \begin{equation*} 
    E \ =\ \sum_{j=1}^n E_j \ =\ \sum_{j=1}^n v_j
    \cdot W_j^{(\alpha-1)/\alpha} \cdot
    \left(\frac{1}{(\alpha-1)\lambda}\right)^{(\alpha-1)/\alpha} =\ \gamma
    \cdot \left(\frac{1}{(\alpha-1)\lambda}\right)^{(\alpha-1)/\alpha} \,.
  \end{equation*} 
  Then, we can express the values~$E_j$
  in~\eqref{eq:E_j} independently of~$\lambda$ and conclude that $E_j = E \cdot v_j \cdot  W_j^{\frac{\alpha-1}{\alpha}}/\gamma.$
\end{proof}

Using this optimal energy assignment (Theorem~\ref{thm:E_j}), the scheduling problem at hand reduces to finding the permutation $\pi$ that minimizes
\begin{align} 
    \sum_{j=1}^n w_j C_j(E)  &=  \sum_{j=1}^n W_j^{\pi} \cdot
\left(\frac{v_j^{\alpha}}{E_j}\right)^{\frac{1}{\alpha-1}}  = 
\frac{1}{E^{\frac{1}{\alpha-1}}} \cdot \left(\sum_{j=1}^n  v_j \cdot 
\left(W_j^{\pi}\right)^{\frac{\alpha-1}{\alpha}}\right)^{\frac{\alpha}{\alpha-1}
} \label{eq:scheduling-cost}\,,
  \end{align}
where the last equation comes from the definition of $\gamma_{\pi}$ (see Theorem~\ref{thm:E_j}) and standard transformations. Interestingly, the optimal job sequence is independent of the energy
distribution, and furthermore it is independent of the overall energy budget. In other words, one scheduling sequence is universally optimal for all energy budgets. As we will see this sequence is obtained by solving~{\em in~weight-space} a~(standard) scheduling problem 
with a cost function that depends on the power function. A similar observation
was independently made by V\'asquez~\cite{vas12}.

\begin{theorem}\label{thm:opt-sequence} 
Given a power function $P(s)=s^{\alpha}$, there is a universal sequence that
minimizes $\sum_j w_j C_j$ for any energy budget. The sequence is given by
reversing an optimal solution of the scheduling problem~\abc{1}{}{\sum w_j
C_j^{(\alpha-1)/\alpha}} (on a single machine of unit speed).
\end{theorem}
\begin{proof} 
  Equation~\eqref{eq:scheduling-cost} implies that the optimal job sequence is independent of the available energy budget~$E$  since it only plays a role in the factor outside the sum, which is independent of the permutation. Since the exponent $\alpha/(\alpha-1)$ is constant, the problem of finding the optimal sequence under an optimal energy-distribution reduces to finding the sequence that minimizes 
$$\sum_{j=1}^n  v_j \cdot \left(W_j^{\pi}\right)^{(\alpha-1)/\alpha}.$$
Now recall the reinterpretation that the 2D-Gantt chart view offers (see
Section~\ref{sec:weight-space}). Then $W_j^{\pi}$ is the completion weight of job
$j$ in a schedule that follows sequence $\pi$ in time-space (and the reverse
order in weight-space). We conclude that this problem is equivalent to the
scheduling problem in {\em weight-space} with {\em varying speed on the
weight-axis or general cost function in the weight-space}. This problem can be
directly translated into minimizing the total weighted completion time on a
machine with varying speed (or the desired form with a generalized cost
function) by re-interpreting
weight-space~as~time-space. 
We simply define
a new problem in time-space with processing times $v'_j=w_j$ and
weight~$w_j'=v_j$, where the objective is to find a permutation of jobs minimizing $\sum_{j=1}^n w_{\pi(j)}'
\cdot f\left(\sum_{k=1}^j v'_{\pi(k)}\right)$, for $f:x\rightarrow
 x^{(\alpha-1)/\alpha}$. This is a problem of the
desired type. By Section~\ref{sec:weight-space} it is easy to see that a
solution $\pi'$ to the new problem in time-space, has a corresponding solution~$\pi$ in the weight-space with same total cost; $\pi$ is the reverse of~$\pi'$. 
\end{proof}

Thus, the scheduling part of the speed-scaling scheduling problem reduces to a
problem which can be solved by our PTAS from Section~\ref{sec:PTAS-given-speeds}.
Since the cost function~$f(x)= x^{(\alpha-1)/\alpha}$ is concave for~$\alpha>1$,
the specialized PTAS in~\cite{stillerW10} also solves it. Combining
Theorems~\ref{thm:E_j} 
and~\ref{thm:opt-sequence} 
gives the main result.



\begin{theorem}\label{thm:continuous-PTAS}
  Let $\alpha\ge1$ be a (constant) rational number. There is a PTAS for the continuous speed-scaling and scheduling problem with
a given energy budget~$E$ for continuous speed and power function $P(s)=s^{\alpha}$. Indexing jobs in this order,
the $(1+\eps)$-approximate pareto curve describing the approximate scheduling
cost as a function of the available energy is given by the right-hand-side of
Equation~\eqref{eq:scheduling-cost}.

\end{theorem} 
\begin{proof}
 The previous theorem argues that our energy problem is equivalent to~\abc{1}{}{\sum w_j
C_j^{(\alpha-1)/\alpha}} in terms of optimal solutions. However, approximation factors are not exactly preserved: as can be seen from Equation~\eqref{eq:scheduling-cost}, a solution with cost $Z$ for~\abc{1}{}{\sum w_jC_j^{(\alpha-1)/\alpha}} corresponds to a solution of cost $Z^{\frac{\alpha}{\alpha-1}}$ for the speed-scaling problem. Hence, a $\beta$-approximation algorithm for the static-speed problem yields an approximation factor of $\beta^{\alpha/(\alpha-1)}$ for the dynamic-speed problem. Since $\alpha\ge1$ is a constant, taking $\beta=(1+\e)$ yields an approximation factor of $(1+\e)^{\alpha/(\alpha-1)}= 1+\OO(\e)$ for the speed-scaling problem (for small enough $\e>0$). Therefore it suffices to give a PTAS for~\abc{1}{}{\sum w_j C_j^{(\alpha-1)/\alpha}}. 

To apply Theorem~\ref{thm:PTAS-given-speeds} it suffices to specify the oracle function $f$. In our case $f(x)=x^{(\alpha-1)/\alpha}$ might yield irrational numbers. However, since we aim for a PTAS it suffices to define a polynomial time oracle $\tilde{f}$ that approximates $f$ within a $1+\e$ factor. This can be done with standard techniques from numerical analysis, e.g., Newton's method~\cite{StoerBulirsch}.
\end{proof}


\section{Speed-scaling for discrete speeds}
\label{sec:discrete-speed}

In this section we consider a more realistic setting, where the machine speed can be chosen from a set of $\ka$ different speeds~$s_1 > \ldots > s_{\ka}>0$. We also allow to run the machine at zero speed, which we assume to induce zero power consumption. For this problem we resolve the complexity status and show that it is NP-hard even when $\ka=2$. For arbitrarily many speed states we give a PTAS, and if $\ka$ is constant an FPTAS. 


\subsection{A PTAS for discrete speeds}
\label{sec:ptas-discrete}

To derive our algorithm, we adapt the PTAS for scheduling on a machine with given varying speed (Section~\ref{sec:PTAS-given-speeds}) and incorporate the allocation of energy. Fortunately, many of the techniques to derive that PTAS, in particular the computation of sets $\mathcal{\tilde{F}}_u$, are independent of the speed of the machine. Thus we can use them without modifications. 


Consider the power function~$P(s)$ to be an arbitrary computable function. We adopt the same definitions of weight intervals $I_u$ and sets $\mathcal{F}_u$ as in Section~\ref{sec:PTAS-given-speeds}. For a subset of jobs~$S\in \mathcal{F}_u$ and a value $z\geq 0$, let $E[u,S,z]$ be the minimum total energy necessary for scheduling~$S$ such that all completion weights are in interval~$I_u$ or before 
and the scheduling cost is at most~$z$, i.e., $\sum_{j\in S} x_j \cdot C_j^w \leq z$ where $x_j$ is the execution time under some feasible speed assignment. 
The recursive definition of a state is as follows:
\[
E(u,S,z) = \min\{ E(u-1,S',z') + \APPROX_u(S\setminus S', z-z') : S'\in \mathcal{F}_{u-1}, S'\subseteq S\}.
\]

Here~$\APPROX_u(S\setminus S', z-z')$ is the minimum energy necessary for scheduling all jobs $j\in S\setminus S'$ with $C_j^w\in I_u$, such that their partial (rounded) cost $\sum_{j\in S\setminus S'} x_j (1+\e)^u$ is at most $z-z'$.
\begin{lemma}\label{lem:LP-energy}
  The value $\APPROX_u(S\setminus S', z-z')$ can be computed in polynomial time.
\end{lemma}

\begin{proof} 
We set an LP computing $\APPROX_u(S\setminus S', z-z')$. Let the solution variable $\ell_i\geq 0$,~$i\in\{1,\ldots,\ka\}$, denote the length of the time interval in which the machine is running at speed~$s_i$. Consider the following LP, 
\begin{align}
  \nonumber
  \min &\sum_{i=1}^{\ka} \ell_i \cdot P(s_i)\\
   &\sum_{i=1}^{\ka} \ell_i \cdot s_i = \sum_{j\in S\setminus S'} v_j \label{eq:LP2}, \\
   &\sum_{i=1}^{\ka} \ell_i \cdot (1+\eps)^u \leq z-z' \label{eq:LP3},\\
   \nonumber 
   &\ell_i \geq 0.
\end{align}
Here~\eqref{eq:LP2} guarantees that the total processing volume $v(S'\setminus S)$ can be 
completed, and~\eqref{eq:LP3} that the total scheduling cost does not exceed~$z-z'$.
\end{proof}

We let the DP fill the table for~$u\in \{0,\ldots,\nu\}$ with $\nu=\lceil \log_{1+\e} \sum_{j\in J} w_j \rceil$ and $z\in [1,z_{\UB} ]$ for some upper bound such as $z_{\UB}=\sum_{j\in J} w_j \sum_{k=1}^j v_j/s_{\ka}$. Then among all end states $[\nu, J, \,\cdot\,]$ with value at most the energy budget~$E$ we choose the one with minimum cost~$z$. Then we obtain the corresponding $(1+\eps)$-approximate solution for energy $E$ by backtracking.

This DP has an exponential number of entries. However, we can apply results from Section~\ref{sec:PTAS-given-speeds} and standard rounding techniques to reduce the running time.

\begin{theorem}\label{thm:ptas-discrete}
  There is an efficient PTAS for minimizing the total scheduling cost for speed-scaling with a given energy budget.
\end{theorem}
\begin{proof}
  The DP computes a $(1+\eps)$-approximation in exponential time. In Lemma~\ref{lm:compactFu}, we showed how to reduce the exponential number of subsets in~$\mathcal{F}_u$ to a polynomial number at the cost of a factor $1+\OO(\eps)$ in the total scheduling cost. Recall that the sets $\tilde{\mathcal{F}}_u$ given by that lemma are independent of the speed of the machine. Therefore we can use these sets directly in our setting. 
  
  It remains to reduce the number of possible values of cost $z\in [0,z_{\UB}]$. At the cost of a factor $1+\eps$, we may round up in each state the scheduling cost to next integer power of $1+\delta$ with $\delta= (1+\eps)^{1/\nu}-1$. In each state transition of the DP, we loose up to a factor $1+\delta$ in the scheduling cost, which amounts to at most a factor $(1+\delta)^\nu = 1+\eps$ under~$\nu$ state transitions. When restricting to powers of $1+\delta$ then the number of different values in $z\in [0,z_{\UB}]$ is bounded by $\OO(\log_{1+\delta}z_{\UB})  = 
\OO(\nu\cdot \log z_{\UB} / \eps )$. Thus, the number of states in the table is polynomial. We conclude that the algorithm runs in polynomial time.
\end{proof}

\subsection{Speed-scaling with discrete speeds is NP-hard}

We show that speed-scaling for discrete speeds is NP-hard. We provide a reduction based on the problem of minimizing the total weighted tardiness of jobs with a common due date, $1|d_j=d|\sum w_j T_j$, which is known to be NP-hard~\cite{yuan92}. Here, $T_j=\max\{C_j - d, 0\}$ denotes the tardiness of job $j$. We use the following generalization of this result for our reduction.

\begin{lemma}\label{thm:2convex-hard}
  The problem of minimizing $\sum w_jf(C_j)$ on a single machine of unit speed is NP-hard even when $f$ is increasing, convex and piecewise linear with only one breakpoint.
\end{lemma}
\begin{proof}
  Let $\e\ge0$ and define the cost function
  \[
  f_\e (x) = \begin{cases}
    \e\cdot x  & \text{ if } 0 \le x < d,\\
    x - d + \e d & \text{ if } d \le x.
  \end{cases}
  \]
Note that $T_j = f_0 (C_j)$ is the tardiness of job $j$. Now we show that, for $\e>0$ small enough, minimizing $\sum_j w_j T_j$ is equivalent to minimizing $\sum_j w_j f_{\e}(C_j)$. 

Let $k\in \mathbb{N}$, and assume that $w_j, p_j$ and $d$ are natural numbers for all $j$. It is known that the problem of deciding whether there exists a schedule with $\sum_j w_j T_j\le k$ is NP-hard~\cite{yuan92}. Now notice that
\begin{align*}
  \sum_j w_j f_{\e}(C_j) & = \sum_{j: C_j < d} w_j \e C_j  + \sum_{j: C_j \ge d} (C_j -d + \e d) w_j\\
  & = \e \cdot\left( \sum_{j: C_j < d} w_j C_j + \sum_{j: C_j \ge d} d w_j \right) + \sum_j w_j f_0(C_j). 
\end{align*}
Defining $\e= 1/(d \sum_j w_j)\le 1$ (which can be described with polynomially many bits) we obtain that
\[
0 \le \sum_j w_j f_{\e}(C_j) - \sum_j w_j f_{0}(C_j) = \e \cdot\left( \sum_{j: C_j < d} w_j C_j + \sum_{j: C_j \ge d} d w_j \right) < \e d \sum_j w_j \le 1.
\]
Therefore $\sum_j w_j f_{0}(C_j) \le k$ if and only if $\sum_j w_j f_{\e}(C_j) \le k+1$. We conclude that minimizing $\sum_j w_j f_{\e}(C_j)$ is NP-hard, where $\e\le 1$ is considered as part of the input. 

\end{proof}

Now we can prove the main result.

\begin{theorem}\label{thm:2speeds-hard}
  The problem of minimizing $\sum_j w_jC_j$ on a single machine for discrete speeds is NP-hard, even if the number of available power levels is $2$.
\end{theorem}
\begin{proof}
The problem with $k>2$ speed states can be reduced to the case with $2$ speed states, by adding dummy states of arbitrarily slow speed. Therefore, we prove hardness of the case of two speeds $s_1>s_2$.  

Consider a scheduling instance on a unit-speed processor with the objective of minimizing $\sum_j w_j f_{\e} (C_j)$, where $f_{\e}$ is defined in the proof of the previous lemma. We define an equivalent scheduling instance for minimizing $\sum_j w_j C_j$ on a machine with two possible speed states. In the new instance, the job set is the same and the values $w_j$ and $v_j$ for each job $j$ are also preserved. Let~$s_1=1/\e$ and~$s_2=1$. The total energy budget is~$E=V+d(1/\eps^{\alpha-1}-1)$, where $V$ denotes the total work volume, $\sum_j v_j$. 
A simple interchange argument shows that in an optimal solution the machine runs at decreasing speeds. The time point when the speeds changes is uniquely defined by the energy budget and the total work volume. In this case, the machine runs at speed~$s_1$ until~$\x=\eps d$ and then it runs at speed~$s_2$. Also, the total work volume finished by~$\x$ is~$\x \cdot s_1 = d$.

Consider now a schedule without idle time on a machine with the speed profile just described. Assume that by relabeling the jobs the completion times satisfy that $C_1< C_2 < \ldots < C_n$. Consider scheduling the jobs in a unit speed machine using the same permutation of jobs. In this new schedule the completion times are $C_j' = \sum_{k\le j} v_k$ for all $j$. If it easy to check that $f_\e(C_j')= C_j$. We conclude that the problem of minimizing $\sum_j w_j f_\e(C_j')$ is equivalent to minimizing $\sum_j w_j C_j$ on a machine that has speed $1/\e$ in interval $[0,\e d]$ and speed $1$ afterwards until all jobs are done. By Theorem~\ref{thm:2convex-hard} both problems are NP-hard, wich concludes the proof. 
\end{proof}

\subsection{FPTAS for a constantly many discrete speed-states}
\label{sec:fptases}

Consider the setting where the number of different (non-zero) speeds $\ka$ is constant. We give an FPTAS for this case. Again we will use the dual scheduling view and construct a solution in weight-space. Notice that in this problem setting, jobs may run at more than one speed. We call those jobs {{\em split jobs}. Our approach is as follows: We first use enumeration to determine split jobs, their position in the weight-axis, and the speeds at which they shall run. Then we design an exponential-time dynamic program that fills the remaining jobs running at a single speed into the gaps left between the split jobs. We show then how to reduce the running time of this method to polynomial time by rounding and state-cleaning and loosing only a small factor in the scheduling cost.

Using a standard scaling argument, we may assume w.l.o.g.  that all job weights have integer values.  

\subsubsection{Guessing split jobs and partition of the weight-axis} \label{sec:split-jobs}

Recall that in any optimal solution the speed of the machine is decreasing over time. Thus there are at most $\ka-1$ many split jobs each running at a constant number of different speeds. We show that by restricting the set of possible completion weights to a polynomial size, we may guess in polynomial time the subset of split jobs, the speeds at which each of them is running,  and their completion weights at an affordable loss in the total cost. The placement of split jobs in the weight-axis leads naturally to a partition of the weight-axis into (at most)~$\ka$ intervals to which the remaining non-split jobs shall be assigned.

\begin{lemma}\label{lem:guess-compl-weight}
  By increasing the scheduling cost by at most a factor~$1+\eps$, we may assume that the completion weights of split jobs are integer powers of $1+\beta$ for $\beta= (1+\eps)^{1/n}-1$. 
\end{lemma}
\begin{proof}
This follows by multiplying the completion weight of each job by $1+\beta$ as in the Weight Stretch procedure; see Section~\ref{sec:PTAS-given-speeds}. Each time we do this we can decrease the completion weight of one split job to a integer power of $1+\beta$. This increases the total cost by a factor $1+\beta$ each time, which amounts to at most a factor $(1+\beta)^{\ka-1} < 1+\eps$ for at most~$\ka-1<n$ split jobs. 
\end{proof}


\begin{lemma} \label{lem:guess-split-jobs}
By loosing at most a factor $1+\eps$ in the scheduling cost, we can enumerate in time $\OO(n^{2\ka-2}\cdot \nu^{\ka-1})$, with $\nu=\lceil \log_{1+\e} \sum_{j\in J} w_j \rceil$,  the set of split jobs, the speeds at which they run, and their completion weight.
\end{lemma}
\begin{proof}
  The speed of the machine is decreasing and jobs run non-preemptively. Hence, a split job will run at two or more decreasing speeds $s_i>s_{i+1}>\ldots s_{i'}$ while there is no other job running at speed~$s_k$ with $i<k<i'$. However, not all available speeds might be used.  There are $\OO(n^{\ka-1})$ many choices for selecting the set of (at most)~$\ka-1$ split jobs and the speeds at which each of them is running. 

Given a set of jobs we enumerate all possible completion weights for split jobs. Thereby, we restrict to powers of $1+\beta$ loosing at most a factor $1+\e$ in the cost (Lemma~\ref{lem:guess-compl-weight}). There are $\lceil \log_{1+\beta}\sum_jw_j\rceil = \lceil \log_{(1+\e)^{1/n}}\sum_jw_j\rceil \in \OO(n\cdot \nu)$ many possible completion weights per job. Thus, in total we have to consider $\OO(n^{\ka-1}\cdot (n\nu)^{\ka-1})$
many choices for split jobs with their speeds and positions in the weight-axis.
\end{proof}

Consider a fixed choice for split jobs $j_1,\ldots,j_{\ka-1}$ and their completion weights~$C_{j_1}^w < C_{j_2}^w < \ldots < C_{j_{\ka-1}}^w$. For convenience we add dummy jobs with zero-weight and -work volume if there are less than $\ka-1$ split jobs. The set of $\ka-1$ split jobs partitions the weight-space into~$\ka$ subintervals~$I_1,\ldots,I_{\ka}$ of idle weight between the placed split jobs. More precisely, $I_i=[\w_i,\w_{i+1}-w_{j_i}]$ where $\w_i=C_{j_{i-1}}^w$, for~$i\in\{2,\ldots,\ka\}$, and~$\w_1=0$. Let the last interval~$I_{\ka}$ be bounded from above by $\sum_{j\in J}w_j-\w_{\ka}$. Intervals may also be empty.

To obtain a schedule, we have to fill the remaining jobs non-preemptively in these idle-weight intervals (keeping the split jobs where they are). All jobs in one subinterval will run at the same speed. Again, recall that the speeds are only decreasing in time which means that they are increasing in weight-space.  We simply guess the uniform speed~$s'_i$ associated with~$I_i$ such that $s'_1 \le s'_2\le \ldots \le s'_{\ka}$ in accordance with the speeds of the split jobs between intervals. I.e., each speed~$s$ for a split jobs~$j_i$ separating intervals~$I_i$ and $I_{i+1}$ must satisfy~$s'_i \leq s \leq s'_{i+1}$. Notice that because of the dummy jobs there might be more than one interval with the same speed.

\begin{corollary}
  By losing at most a factor~$1+\eps$ in the scheduling cost we can reduce in time $\OO(\ka^{\ka}\cdot n^{2\ka-2}\cdot \nu^{\ka-1})$ the speed-scaling problem to non-preemptive scheduling in weight-space in a given set of available idle-weight intervals~$I_1,I_2,\ldots,I_{\ka}$ and speed~$s'_i$ for jobs being assigned to~$I_i$.
\end{corollary}
  
  \subsubsection{Dynamic program} \label{sec:DP}
  
We construct a DP that finds a partition of the set of non-split jobs into $\ka$ subsets each of which is assigned to an individual interval~$I_i$. The jobs in each individual set are scheduled according to Reversed Smith rule in weight-space, that is, in non-decreasing order of ratios~$w_j /v_j$. Let all jobs be indexed in this order. 

The dynamic program generates a state~$[k,z,y_1,\ldots,y_{\ka}]$ if there is a feasible
schedule of jobs $1,\ldots,k$, in which the  total weight scheduled in interval~$I_i$~(excluding the split job)  is $y_i$. The total scheduling cost~(including split jobs) is~$z:=\sum_{j=1}^kx_jC_j^w$, with~$x_j=v_j/s'_{i}$ being the execution time of a job~$j$ in interval~$i$.  The value of the state~$[k,z,y_1,\ldots,y_{\ka}]$ is the minimum energy that is necessary for obtaining such a schedule. The dynamic program starts with the states~$[0,z,0,\ldots,0]$. For each~$z$-value a linear program computes the minimum energy that is necessary to obtain this scheduling value when scheduling only the set of split jobs~$J_s$. 
It determines the power assigned to each split job and thus their actual execution times. Let~$\ell_{ji}$ be the amount of time that split job~$j\in J_s$ is running at a valid speed~$s'_i$ (given by Lemma~\ref{lem:guess-split-jobs}).
\begin{align*}
  \min & \sum_{j \in J_s} \sum_{i=1}^{\ka} \ell_{ji} P(s'_i)\\
  & \sum_{j\in J_s} C_j^w \cdot \sum_{i=1}^{\ka} \ell_{ji} \leq z, &\\
  & \sum_{i=1}^{\ka} \ell_{ji} s'_i  = v_j\, \qquad &\text{for all }j\in J_s,\\
  & \ell_{ji} \geq 0 \qquad & \text{for all } j\in J_s, i\in \{1,\ldots,\ka\},\\
  & \ell_{ji} = 0 \qquad & \text{for all } j\in J_s,\, s'_i \textup{ not valid for } j \,.
\end{align*}

After computing the starting states, the DP computes all states by moving from
any state $[j-1,z,y_1,\ldots,y_{\ka}]$ to at most $\ka$ new states~$[j,z',y_1',\ldots,y'_{\ka}]$ by assigning job~$j$ to intervals~$I_i$ for $i\in\{1,\ldots,\ka\}$. Then
\begin{eqnarray}
  z'=z + \frac{v_j}{s'_i}\cdot \left(\w_i+y_i+w_j\right) \   \textup{and}\  y_i' =y_i+w_j \  \textup{and}\  y_{i'}'=y_{i'} \textup{ for } i'\neq i\,, \label{eq:before-k-energy}
\end{eqnarray}
provided that $y_i'\leq |I_i|-w_{j_i}$, where $j_i$ is the $i$th split job. 
The value of the new state is
\begin{equation}
  E[j,z,y_1,\ldots,y_{\ka}] = E[j-1,z,y_1,\ldots,y_{\ka}] + \frac{v_j}{s'_i}\cdot P(s'_i). \label{eq:DP-energy}
\end{equation}
If there exists another state with smaller energy value $E'[j,z,y_1,\ldots,y_{\ka}] < E[j,z,y_1,\ldots,y_{\ka}]$ we discard the new one with larger energy value.

An optimal schedule can be obtained by finding a state $E[n,z,y_1,\ldots,y_{\ka}] \leq E$ with minimum
$z$ and backtracking from that state. Since the $z$-values are bounded by
$z_{UB}:=\sum_{j=1}^nw_j(\sum_{\ell=1}^j v_{\ell}/s_{\ka})$ and the $y_i$-values are bounded by $|I_i|$,
the running time of this dynamic programming algorithm is $\OO(n\cdot z_{UB}\cdot \max_i|I_i|^{\ka})$.

\subsubsection{Rounding}  \label{sec:rounding}

In a fully polynomial-time algorithm, we can neither afford to consider all possible objective values~$z$, nor can we consider all possible $y_i$-values. 

Consider first the number of possible values~$z$ of scheduling cost. We round them the same way as we have done in the PTAS for an arbitrary number of discrete speeds in Theorem~\ref{thm:ptas-discrete}. Given the upper bound on the cost, $z_{\UB}=\sum_{j=1}^nw_j(\sum_{\ell=1}^j v_{\ell}/s_{\ka})$, we can reduce the number of possible values in $z\in [0,z_{\UB}]$ to $\OO(\nu\cdot \log z_{\UB} / \eps) = \OO(\nu\cdot n / \eps^2)$ by restricting to powers of $1+\delta$ with $\delta= (1+\eps)^{1/\nu}-1$ and lose only a factor~$1+\eps$ in the scheduling cost. Recall that $\nu= \lceil \log_{1+\e}\sum_{j\in J} w_j\rceil$. Let~DP$_z$ denote this dynamic program that rounds only the scheduling cost.


We now take care of the $y$-values. The idea is to reduce the number of states
by removing those with the same (rounded) objective value and nearly the same
total weight in all intervals~$I_i$. Among them, we store those that require the minimum amount of energy. To do so, we use the same discretization of the weight-axis as for guessing the completion weights of split jobs (Section~\ref{sec:split-jobs}). When the DP adds a job~$j$ to some interval~$I_i$ and updates the total weight~$y_i'=y_i+w_j$ (see Equation~\eqref{eq:before-k-energy}) then we store only the information on $y_i'$ rounded down to the closest integer power of $1+\beta$, with $\beta = (1+\eps)^{1/n}-1$. Now, among all states with the same rounded values~$z,y_1,\ldots,y_{\ka}$ we store the one with minimum energy consumption. Let~DP$_{z,y}$ denote the modified dynamic program that rounds $z$ and $y$-values.

Rounding down the $y_i$-values will incur an error in the computation of scheduling cost; more precisely, interpreting the solution of DP$_{z,y}$ as a job (weight) assignment to intervals, then the $y$-values stored for describing a DP state underestimate the true weight assigned to an interval, and thus, the DP also underestimates the total scheduling cost~$z$. We have to show in the following that this error is small compared to the true value of a feasible solution. We will also show that the energy consumption computed by the DP corresponds to the exact energy required in a feasible solution.


\begin{lemma}\label{lem:DP-underestimates}
  Suppose that algorithm \textup{DP}$_z$ on an instance with $n$ jobs 
   finds a chain of states\footnote{\emph{Chain} of states means that, for
     $j=0,\dots,n-1$, state $[j+1,z_{j+1}^*,y_{1,j+1}^*,\ldots,y_{\ka,j+1}^*]$ is obtained from
    $[j,z_j^*,y_{1,j}^*,\ldots,y_{\ka,j}^*]$ by adding job $j+1$ according to~\eqref{eq:before-k-energy}.}  $[0,z_0^*,0,\ldots,0], [1,z_1^*,y_{1,1}^*,\ldots ,y_{\ka,1}^*],\ldots,[n,z_n^*,y_{1,n}^*,\ldots ,y_{\ka,n}^*]$. Then the algorithm \textup{DP}$_{z,y}$ finds for each $j\in\{1,\ldots, n\}$ a  state $[j,z_j,y_{1,j},\ldots,y_{\ka,j}]$ of energy value at most $E[j,z_j^*,y_{1,j}^*,\ldots,y_{\ka,j}^*]$ such that 
 \begin{equation}
   y_{i,j} \leq y_{i,j}^*
   \qquad \textup{and} \qquad
    z_j\leq z_j^*\, . \label{eq:underestimates}
  \end{equation}
\end{lemma}
\begin{proof}
 We give a proof by induction on the number of jobs $j$. For $j=0$ the property is clearly true since \textup{DP}$_{z,y}$ and \textup{DP}$_{z}$ have the same starting states. 
  
  Suppose that the lemma is true for~$j$ jobs, and thus \textup{DP}$_{z,y}$ obtains state $[j,z_j,y_{1,j},\ldots,y_{\ka,j}]$ satisfying the properties of the lemma. Now consider state $[j+1,z_{j+1}^*,y_{1,j+1}^*,\ldots ,y_{\ka,j+1}^*]$ that \textup{DP}$_z$ obtained
  from $[j,z_j^*,y_{1,j}^*,\ldots,y_{\ka,j}^*]$ according to~\eqref{eq:before-k-energy} by adding job $j+1$ to interval~$I_i$, for some~$i\in\{1,\ldots,\ka\}$. Similarly, starting from $[j,z_j,y_{1,j},\ldots,y_{\ka,j}]$, Algorithm \textup{DP}$_{z,y}$ considers a state that inserts job $j+1$ to interval~$I_i$. This yields a new state $[j+1,z_{j+1},y_{1,j+1},\ldots,y_{\ka,j+1}]$ that satisfies $z_{j+1}=z_j+v_{j+1}/s'_i  \cdot (\w_i + y_{i,j}+w_{j+1})$ and $y_{i,j+1}$ as $\bar{y}_{i,j+1}=y_{i,j}+w_{j+1}$ rounded down to the nearest power of $1+\beta$, while it keeps~$y_{i',j+1}=y_{i',j}$ for all~$i'\ne i$. 

By inductive hypothesis, we have that $y_{i,j}\leq y^*_{i,j}$ and thus $$z_{j+1}=z_j+v_{j+1}/s'_i  \cdot (\w_i + y_{i,j}+w_{j+1}) \leq z^*_{j+1}.$$
Moreover, since we round down the value $\bar{y}_{i,j+1}$ 
to~$y_{i,j+1}$ we obtain that
 \[
y_{i,j+1} \leq \bar{y}_{i,j+1} = y_{i,j}+w_{j+1} \leq  y^*_{i,j}+w_{j+1} = y^*_{i,j+1}. 
\]

It remains to argue on the value of the state, that is, the energy cost. According to  Equation~\eqref{eq:DP-energy} the value of the state as computed by DP$_{z,y}$ is 
\begin{align*}
  E[j+1,z_{j+1},y_{1,j+1},\ldots,y_{\ka,j+1}] &= E[j,z_j,y_{1,j},\ldots,y_{\ka,j}] + \frac{v_{j+1}}{s'_i}\cdot P(s'_i)\\
& \leq E[j,z^*_j,y^*_{1,j},\ldots,y^*_{\ka,j}] + \frac{v_{j+1}}{s'_i}\cdot P(s'_i)\\
& =  E[j+1,z^*_{j+1},y^*_{1,j+1},\ldots,y^*_{\ka,j+1}].
\end{align*} 

We cannot guarantee that state $[j+1,z_{j+1},y_{1,j+1},\ldots,y_{\ka,j+1}] $survives. But in case it does not then we have found another partial solution with the same objective value~$z_{j+1}$, the same values $y_{i,j+1}$, and an even smaller state value (energy). This concludes the lemma.
\end{proof}

The Algorithm DP$_{z,y}$ computes an assignment of jobs to weight intervals but it underestimates the total weight assigned to an interval and thus the scheduling cost. We show that the true scheduling cost when scheduling according to the solution found by DP$_{z,y}$ is bounded.

\begin{lemma}\label{lem:feasible-DP}
   Suppose that algorithm \textup{DP}$_{z,y}$ on an instance with $n$ jobs 
   finds a chain of states $[0,z_0^*,0,\ldots,0], [1,z_1^*,y_{1,1}^*,\ldots ,y_{\ka,1}^*],\ldots,[n,z_n^*,y_{1,n}^*,\ldots ,y_{\ka,n}^*]$.
 Then for each state $[j,z_j^*,y_{1,j}^*,\ldots ,y_{\ka,j}^*]$, with $j\in\{1,\ldots, n\}$, there exists a feasible partial schedule of split jobs and jobs~$1,\ldots,j$ using an energy budget of at most~$E[j,z_j^*,y_{1,j}^*,\ldots ,y_{\ka,j}^*]$. Moreover, if $y_{i,j}$ denotes the total weight of jobs assigned to interval $I_i$ in the partial schedule and $z_j$ is the scheduling cost, then
 \begin{equation}
   y_{i,j} \leq (1+\beta)^{j} \cdot y_{i,j}^*
   \qquad \textup{and} \qquad
    z_j\leq (1+\beta)^j\cdot z_j^*\, . \label{eq:underestimates2}
  \end{equation}
\end{lemma}
\begin{proof}
We give a proof by induction on $j$. By definition of the starting state $[0,z_0^*,0,\ldots,0]$ there exists a partial schedule of the split jobs with cost at most $z_0^*$. Thus the base case of the induction follows.


For a given $j$, assume that the DP obtains state $[j+1,z_{j+1}^*,y_{1,j+1}^*,\ldots ,y_{\ka,j+1}^*]$ by adding job $j+1$ to interval $I_i$. By induction hypothesis suppose that there exists a partial schedule satisfying the claim for jobs $1,\ldots,j$. We construct the new schedule for jobs $1,\ldots,j+1$ by also adding $j+1$ to $I_i$. The total weight assigned to interval~$I_i$ in this solution is \jcom{I rephrased the last two paragr.}
\begin{align*}
  y_{i,j+1} = y_{i,j}+w_{j+1} \leq  (1+\beta)^j \cdot y^*_{i,j}+w_{j+1} \leq   (1+\beta)^j \cdot (y^*_{i,j}+w_{j+1}). 
\end{align*}
Since DP$_{z,y}$ rounds down the~$y$-value to the next integral power of $1+\beta$, we have that 
$$y^*_{i,j+1} \geq \frac{1}{1+\beta}\cdot \left(  y^*_{i,j} + w_{j+1} \right).$$ And thus we conclude~$y_{i,j+1}\leq  (1+\beta)^{j+1} \cdot y_{i,j+1}$.

Consider now the total scheduling cost of the feasible schedule after adding job~$j+1$. In principle it consists of the scheduling cost~$z_j$ before adding job~$j+1$ plus the cost for the new job. However there is an possible extra source for error. Since the DP rounded down~$y$-values, we cannot guarantee that the total weight assigned to an interval~$I_i$ actually fits into this interval. (Recall, that these intervals are defined by the placement of the split jobs in the weight-axis which is in principle flexible.) Thus, we may increase the completion weight of already assigned jobs by at most a factor~$1+\beta$ which means increasing $z_j$ by this factor. Then, by again using the induction hypothesis and the already proven first condition in~\eqref{eq:underestimates2} we get
\begin{align*}
  z_{j+1} &\leq (1+\beta)\cdot z_j+v_{j+1}/s'_i  \cdot (\w_i + y_{i,j}+w_{j+1})\\
  & \leq  (1+\beta)^{j+1} \cdot z^*_j + v_{j+1}/s'_i  \cdot (\w_i +  (1+\beta)^j\cdot y^*_{i,j}+w_{j+1}) \\
&\leq (1+\beta)^{j+1}\cdot \left( z^*_j+v_{j+1}/s'_i  \cdot (\w_i + y^*_{i,j}+w_{j+1})\right) = (1+\beta)^{j+1}\cdot z^*_{j+1}.
\end{align*}

Concerning the energy estimation, recall that the DP determined the energy cost precisely according to  Equation~\eqref{eq:DP-energy}. Thus, an inductive argument shows that the constructed schedule incurs into the same energy consumption.
\end{proof}

Now we can prove the main result.

\begin{theorem}\label{thm:fptas-discrete}
  There is an FPTAS for speed-scaling with a given energy budget for $\min \sum w_jC_j$ on a single machine with constantly many discrete speeds. 
\end{theorem}
\begin{proof}
  The FPTAS is as follows: We guess the split jobs, their speeds and positions which gives us a partition of the weight-space into $\ka$ idle-weight intervals (see Section~\ref{sec:split-jobs}). Then we run DP$_{z,y}$ and take as final solution the assignment of jobs to intervals that the DP computes. 

  Let $\OPT$ denote the scheduling cost of an optimal solution, and let~$z(A)$ denote the scheduling cost of a solution computed by algorithm~$A$. 
Lemma~\ref{lem:DP-underestimates} guarantees that $\textup{DP}_{z,y}$ finds a final state of cost $z(\textup{DP}_{z,y})\leq z(\textup{DP}_z)$. We can argue that $z(\textup{DP}_z)\leq (1+\e)^2 \OPT$ because we lose one factor $1+\eps$ when guessing the split jobs (Lemma~\ref{lem:guess-split-jobs}) and another factor $1+\e$ when rounding the~$z$-values in $\textup{DP}_z$. Taking the assignment of jobs to intervals as computed by $\textup{DP}_{z,y}$, we obtain a feasible scheduling solution of cost
$z_n \leq (1+\beta)^n z(\textup{DP}_{z,y}) \leq (1+\eps) z(\textup{DP}_{z,y})$, where~$\beta=(1+\eps)^{1/n}-1$ (Lemma~\ref{lem:feasible-DP}). Thus, we find a feasible solution of scheduling cost at most~$(1+\e)^3\OPT$. 

Furthermore, Lemmas~\ref{lem:DP-underestimates} and~\ref{lem:feasible-DP} guarantee that our final solution uses as much energy as an optimal solution. Thus we stay within the energy bound. 

It remains to show that the running time is polynomial in the input and~$1/\eps$. By Lemma~\ref{lem:guess-split-jobs} the enumeration step leading to the partitioning of the weight-axis takes time $\OO(\ka^{\ka}\cdot n^{2\ka-2} \cdot \nu^{\ka-1})$ with $\nu=\lceil \log_{1+\e} \sum_{j\in J} w_j \rceil$. The original (exponential time) dynamic program runs at time~$\OO(n\cdot z_{UB} \cdot \max_i |I_i|)$ (see Section~\ref{sec:DP}). Algorithm~$\textup{DP}_{z,y}$ rounds the $z$- and $y$-values and with the argumentation in Section~\ref{sec:rounding} it  thus  runs in time~$\OO(n\cdot (\nu \cdot n/\eps^2) \cdot (n \cdot \nu) )=\OO(n^3/\eps^2\cdot \nu^2)$. Since we run the DP for each guess of split jobs, we obtain a total running time~$\OO(\ka^{\ka}\cdot n^{2\ka+1} /\eps^2\cdot \nu^{\ka+1})$ which is polynomial in the input and $1/\eps$.
\end{proof}

\section{Speed-scaling with release dates on multiple machines}
\label{sec:parallelMachines}

We can use our results obtained in the dynamic-speed setting  to approximate the more general problem of preemptively scheduling jobs with non-trivial release dates on identical parallel machines.  We use the fact that we can handle jobs without release dates on a single machine and apply a {\em fast single machine relaxation}~\cite{chekuriMNS01}. For the relaxation, we assume that we have a single machine that is $m$ time faster than one of the original machines: at a power level $p$ the single machine runs at speed $m\cdot p^{1/\alpha}$, while at the same power level one of the original machines runs at speed $p^{1/\alpha}$. Thus, if an amount of energy $E_j$ for job $j$ implies a execution time of $x_j$ on an original machine, then the same energy implies an execution time of $x_j/m$ on the fast single machine.

After using our PTAS to solve the single machine relaxation without release dates, we keep the energy assignments~$E_j$ computed in the relaxation and apply standard {\em preemptive list scheduling} on parallel machines respecting release dates. However, the difficulty lies in bounding the actual execution times~$x_j$ in our final solution, since we do not have any information about the optimal execution times~$x_j^*$.

The trick we use is as follows: Suppose we knew the total weighted value of execution times in an optimal schedule~$\sum_{j\in J} w_jx^*_j = X^*$. Then it is easy to verify that the fast single machine relaxation with the {\em additional constraint} $\sum_{j\in J} w_j m x^1_j \leq X^*$ on the weighted actual executions times~$x_j^1$ on the fast machine still gives a lower bound.  Consider the problem of scheduling a job set~$J$ (with release dates) on $m$ parallel machines using an energy budget~$E$. Let $Z(X^*)$ be the cost 
of an optimal schedule using energy~$E$ and $\sum_{j\in J} w_jx^*_j = X^*$. Consider an optimal preemptive schedule with cost $Z_1(X^*)$ for $J$ without release dates on a single machine of speed~$m$ with energy~$E$ and the additional constraint $\sum_{j\in J} w_j m x^1_j \leq X^*$. 
\begin{lemma}\label{lem:fast-relax}
$Z_1(X^*) \leq Z(X^*)$.
\end{lemma}
\begin{proof}
  The proof goes along the same lines as in the non-energy setting in~\cite{chekuriMNS01}. Using time discretization, any parallel machine schedule can be converted into a feasible preemptive schedule on a fast single machine without increasing the total cost and without changing the total energy given to each job. Thus, an optimal single machine schedule gives a lower bound. 
\end{proof}

Given $X^*$, we can solve the restricted fast single machine relaxation using the PTAS from Theorem~\ref{thm:continuous-PTAS} (continuous speeds) or Theorem~\ref{thm:ptas-discrete} (discrete speeds), respectively. We can directly implement the additional constraint of restricting the total weighted execution time by adding an entry to the corresponding dynamic programming table which tracks this value for each partial solution. To guarantee polynomial running time, we round the values to powers of $1+\eps$ at the cost of an additional factor $1+\eps$ in approximation guarantee.

The solution of the fast single machine relaxation gives a priority ordering for the preemptive list scheduling algorithm to obtain the final parallel machine solution. It remains the issue, that we do not know $X^*$. Essentially, we run the algorithm (fast single machine relaxation plus preemptive list scheduling) for every possible value $X^* \in [X_L,X_U]$, for some upper and lower bounds $X_L,X_U$ that we define below, and we pick the best feasible solution. Again, to guarantee a polynomial running time we choose only values that are powers of $1+\eps$ at the cost of a small increase in the approximation guarantee.

A simple lower bound on~$X^*$ is obtained by giving each job the maximum amount of energy~$E$. Recall that $x_j=(v_j^{\alpha}/E_j)^{1/(\alpha-1)}$. Thus,
\begin{align*}
  X^* = \sum_{j\in J} w_jx^*_j \geq \sum_{j\in J} w_j \left(\frac{v_j^{\alpha}}{E} \right)^{\frac{1}{\alpha-1}} =: X_L\,.
\end{align*}
An upper bound can be obtained as follows: the optimal execution times~$x_j^*$ are bounded by the completion times in an optimal solution, and thus, $X^* \leq \opt$. The value~$\opt$ obtained on multiple machines is not larger than the optimal solution for the same job set and energy using just a single machine. Now, for the cost of an optimal single machine solution we gave an explicit expression in Equation~\eqref{eq:scheduling-cost}. This expression used a solution-dependent remaining weight parameter $W_j$ which we crudely bound by $n\cdot w_{max}$, with $w_{\max}:=\max_{j\in J}w_j$. We obtain
\begin{align*}
  X^* & \leq \opt \leq \frac{1}{E^{\frac{1}{\alpha-1}}} \cdot \left(\sum_{j=1}^n  v_j \cdot 
(n \cdot w_{\max})^{\frac{\alpha-1}{\alpha}}\right)^{\frac{\alpha}{\alpha-1}} \\
&=\frac{n \cdot w_{\max}}{E^{\frac{1}{\alpha-1}}} \cdot \left(\sum_{j=1}^n  v_j \right)^{\frac{\alpha}{\alpha-1}} 
=: X_U\,.
\end{align*}

A summary of the algorithm is given below.

\begin{center}
  \begin{minipage}{\textwidth}
    \rule{\textwidth}{0.5pt}\\
    {\sc \bf Algorithm Fast-Relax+List-Scheduling}\\[0.6ex]
    Let $\eps':=\eps/2$.
   For $i=0$ to $\lceil \log_{1+\e'} X_U/X_L \rceil$ do
    \vspace{-1.3ex}
    \begin{enumerate}
      \addtolength{\itemsep}{0.5ex}
      \item Let $X=(1+\eps')^i$.
      \item Compute an energy assignment $E_j$ and a scheduling solution $\pi$ for the given job set~$J$ with release dates set to~$0$ on a single machine running~$m$ times as fast as the original machines, with energy budget $E$, and respecting the additional constraint that $\sum_{j\in J} w_j \cdot m v_j/s_j \leq X$. If there is no solution, then $i\leftarrow i+1$.
    \item Keep the energy assignment and apply preemptive list scheduling according to $\pi$ on $m$ machines respecting release dates, i.e., run at any time the $m$ jobs with the highest priority in~$\pi$ among the available (released, unfinished) jobs. 
    \item If the total cost of this solution is less than previous solutions then keep it, otherwise disregard. 
    \item $i\leftarrow i+1$.
    \end{enumerate}
    \vspace*{-2ex}
    \rule[1ex]{\textwidth}{0.5pt}
  \end{minipage}
\end{center}

\begin{theorem}
  \textup{Fast-Relax+List-Scheduling} is a factor $2+\eps$ approximation for continuous and discrete speed-scaling when jobs have individual release dates.
\end{theorem}
\begin{proof}
  Let~$X^*$ be the total weighted execution time in an optimal parallel machine schedule with cost~$\opt$. Let $\eps':=\eps/2$, and let~$X'$ satisfy $X^* \leq X' \leq (1+ \eps')X^*$. The algorithm returns the minimum cost solution over all weighted completion time bounds~$X$. Thus, the cost of the final solution is bounded by the total cost of the solution obtained based on~$X'$. We show that the cost of this solution is at most $2(1+\eps')\opt = (2+\eps)\opt$.

  Let $Z_1(X)$ denote the cost of an optimal solution to the fast single machine problem with  imposed constraint~$\sum_{j\in J} w_j \cdot m v_j/s_j \leq X$. Clearly, $Z_1(X')\leq Z_1(X^*)$. Let~$C_j^1(X')$ denote the completion time of job $j$ in a solution to the the fast single machine problem with imposed constraint~$X'$ when applying a PTAS~(Theorem~\ref{thm:continuous-PTAS} for continuous speeds or Theorem~\ref{thm:ptas-discrete} for discrete speeds, respectively). Lemma~\ref{lem:fast-relax} and the observation above imply 
$$\sum_{j\in J} w_jC_j^1(X') \leq (1+\eps') Z_1(X') \leq (1+\eps') Z_1(X^*) \leq (1+\eps')\opt\,.$$ 

Now consider the final list scheduling solution obtained for bound~$X'$, and let~$C_j$ denote the completion time of a job~$j$. Recall that the algorithm keeps the energy assignment from the fast single machine relaxation; thus, the execution time of a job~$j$ is~$x_j = m\cdot x_j^1$, where $x_j^1$ is the actual execution time of~$j$ on the fast single machine. By construction, a job~$j$ starts only processing when the first machine becomes available after its release date and after starting all jobs~$k$ with higher priority in~$\pi$ (denoted by $k<_{\pi} j$). Thus, its completion time is bounded by $C_j \leq r_j + \sum_{k<_{\pi} j} x_k/m + x_j$. Thus, the total cost of the algorithms solution~$\alg$ is
\begin{align*} 
  \alg & \leq \sum_{j\in J} w_j r_j + \sum_{j\in J} w_j \sum_{k<_{\pi} j} x_j^1 + \sum_{j\in J} w_j x_j\\
  &\leq  \sum_{j\in J} w_j r_j + \sum_{j\in J} w_j C_j^1(X') + \sum_{j\in J} w_j \cdot m x_j^1\\
  & \leq \sum_{j\in J} w_j r_j + (1+\eps') \opt + \sum_{j\in J} w_j \cdot m x_j^1.
\end{align*}

Now, by construction we have that~$\sum_{j\in J} w_j \cdot m x_j^1 \leq X' \leq (1+\eps')X^*$. Using, the obvious lower bound $\opt\geq \sum_{j\in J} w_j r_j + X^*$, we conclude $\alg \leq 2(1+\eps')\opt$.
\end{proof}

\section{Conclusion}

In this paper we have demonstrated the power of a dual scheduling view for minimizing the total weighted completion time---in particular, when scheduling on a machine that may change its speed. Instead of the standard approach of scheduling along the time-axis, we schedule jobs in the weight-axis of  the well-known two-dimensional Gantt-chart. This change of concept allows to handle the complexity of machine speed changes. We give several algorithms relying on dual techniques and show that they guarantee nearly optimal solutions. 
Most of our results are best possible in terms of approximation guarantees. 

An interesting open question is how to incorporate release dates for the varying-speed scenario 
and improve the~$(4+\eps)$-approximation in~\cite{epsteinLMMMSS12}. While with our current technique we can almost fully resort to the weight-space, release dates would require maintaining a correspondence between weight- and time-space. 

The most challenging open problem in this context concerns min-sum scheduling when each job may have its own non-decreasing cost function~$f_j$. Any improvement of the recent~$4$-approximation~\cite{cheungS11,MestreV14} for~\abc{1}{}{\sum f_j} would be of significant interest. Our PTAS on a machine of varying speed translates into the equivalent setting of scheduling on a unit-speed machine to minimize a general global cost function~$\sum w_jf(C_j)$ and thus give a tight result for this case.

\bibliographystyle{abbrv}
\bibliography{speed}

\end{document}